\documentclass[journal]{IEEEtran}

\ifCLASSINFOpdf
\else
   \usepackage[dvips]{graphicx}
\fi
\usepackage{url}

\usepackage{verbatim}

\usepackage{color}
\usepackage{subfigure}
\usepackage{amsmath,lipsum}

\usepackage{stfloats}

\usepackage{bm}

\usepackage{graphicx}

\usepackage{float}

\usepackage{amssymb}

\usepackage{tabularray}

\newtheorem{theorem}{Theorem}

\newtheorem{lemma}{Lemma}

\newtheorem{proposition}{Proposition}

\newtheorem{corollary}{Corollary}

\newtheorem{remark}{Remark}

\usepackage[noend]{algpseudocode}

\usepackage{algorithmicx,algorithm}

\newenvironment{proof}{{\indent \indent \it Proof:}}{\hfill $\blacksquare$\par}

\usepackage{tabularray}
\usepackage{caption}

\usepackage{etoolbox}  %to don't print the caption in small caps.
\makeatletter
\patchcmd{\@makecaption}
  {\scshape}
  {}
  {}
  {}
\makeatother

\begin{document}
\title{Spatially Correlated RIS-Aided Secure Massive MIMO Under CSI and Hardware Imperfections}

\author{
Dan Yang, Jindan Xu, \IEEEmembership{Member, IEEE,} Wei Xu, \IEEEmembership{Senior Member, IEEE,} Bin Sheng, \IEEEmembership{Member, IEEE,}\\
Xiaohu You, \IEEEmembership{Fellow, IEEE,} Chau Yuen, \IEEEmembership{Fellow, IEEE,} and Marco Di Renzo, \IEEEmembership{Fellow, IEEE}

\thanks{
D. Yang, W. Xu, B. Sheng, and X. You are with the National Mobile Communications Research Lab, Southeast University, Nanjing 210096, China, and also with Purple Mountain Laboratories, Nanjing 211111, China (e-mail: \{dyang, wxu, sbdtt, xhyu\}@seu.edu.cn).

J. Xu and C. Yuen are with the School of Electrical and Electronics Engineering, Nanyang Technological University, Singapore 639798, Singapore (e-mail: jindan.xu@ntu.edu.sg, chau.yuen@ntu.edu.sg).

M. Di Renzo is with Universit\'{e} Paris-Saclay, CNRS, CentraleSup\'{e}lec, Laboratoire des Signaux et Syst\`{e}mes, 3 Rue Joliot-Curie, 91192 Gif-sur-Yvette, France (email: marco.di-renzo@universite-paris-saclay.fr).

Part of this paper was presented at the IEEE 98th Vehicular Technology Conference (VTC), Hong Kong, China, October 2023 [1].}
}
\maketitle
\begin{abstract}
This paper investigates the integration of a reconfigurable intelligent surface (RIS) into a secure multiuser massive multiple-input multiple-output (MIMO) system in the presence of transceiver hardware impairments (HWI), imperfect channel state information (CSI), and spatially correlated channels. We first introduce a linear minimum-mean-square error estimation algorithm for the aggregate channel by considering the impact of transceiver HWI and RIS phase-shift errors. Then, we derive a lower bound for the achievable ergodic secrecy rate in the presence of a multi-antenna eavesdropper when artificial noise (AN) is employed at the base station (BS). In addition, the obtained expressions of the ergodic secrecy rate are further simplified in some noteworthy special cases to obtain valuable insights. To counteract the effects of HWI, we present a power allocation optimization strategy between the confidential signals and AN, which admits a fixed-point equation solution. Our analysis reveals that a non-zero ergodic secrecy rate is preserved if the total transmit power decreases no faster than $1/N$, where $N$ is the number of RIS elements. Moreover, the ergodic secrecy rate grows logarithmically with the number of BS antennas $M$ and approaches a certain limit in the asymptotic regime $N\rightarrow\infty$. Simulation results are provided to verify the derived analytical results. They reveal the impact of key design parameters on the secrecy rate. It is shown that, with the proposed power allocation strategy, the secrecy rate loss due to HWI can be counteracted by increasing the number of low-cost RIS elements. 
\end{abstract}

\setlength{\parskip}{0\baselineskip}

\section{Introduction}
Massive multiple-input multiple-output (MIMO) is a critical technology to support wireless networks for numerous emerging applications [2], [3]. However, massive MIMO requires the deployment of expensive and energy-consuming hardware, such as dedicated power amplifiers and high-resolution digital-to-analog converters (DACs). With the rapid development of metamaterials and millimeter-wave systems on chips (SoC), reconfigurable intelligent surface (RIS) has attracted much attention as a promising technology in the field of wireless communications [4], [5]. Specifically, an RIS is composed of numerous inexpensive nearly passive scattering elements that modify the amplitude and phase of electromagnetic waves and alter the environment in which the waves propagate [6]--[8]. In light of these attractive features, fundamental limits and transmission design have been investigated in [9], which revealed that RIS-assisted massive MIMO offers substantial performance gains.

Thanks to the capability of shaping the propagation characteristics of wireless channels, an RIS can be naturally applied to improve the physical layer security (PLS) by performing nearly passive beamforming. Conventional PLS techniques utilize massive MIMO to focus radio-frequency (RF) signals toward the legitimate users through beamforming, while simultaneously minimize the signal power to the eavesdroppers (Eves) [10], [11]. However, in passive eavesdropping scenarios, the channel state information (CSI) of Eves is rarely available at the base station (BS) and the secrecy performance deteriorates substantially when the channels of the legitimate users and Eves are strongly correlated [12], [13]. Benefiting from the ability to reconfigure the wireless propagation environments, an RIS can be utilized for decorrelating these channels, resulting in improved secrecy performance. 

Recently, RIS has triggered upswing research interest in improving the security in wireless communications [14]--[17]. In [14], the transmit and RIS beamforming were jointly designed by means of artificial noise (AN). It was proved that AN enhances the secrecy rate. The authors of [15] designed secure communications by tackling a multiuser power minimization problem at the transmitter. In [16], in addition, the sum secrecy rate maximization problem was solved by taking multiple single-antenna Eves into account. In [17], a transmission scheme was developed for secure communications by considering outage probabilistic constraints. Apart from the aforementioned representative algorithms for RIS design, the secrecy performance gain that an RIS can provide has recently been characterized from the theoretical standpoint [18]--[23]. In particular, the ergodic secrecy rate in the presence of multiple single-antenna Eves was analyzed in [18]. By using a stochastic geometry method, the secrecy outage probability was then investigated with randomly distributed Eves [19]. The authors of [20] considered the performance limits of an RIS-aided MIMO wiretap channel by exploiting random matrix theory. Regarding RIS-aided massive MIMO systems, the secrecy performance under spatially independent and correlated channels were investigated in [21] and [22], respectively. Further work in [23] designed detectors for RIS-aided massive MIMO systems to counteract active pilot contamination attacks.

Although the secrecy performance can be significantly enhanced by utilizing an RIS, the above-mentioned works considered the assumption that the transceiver is equipped with perfect hardware. In practice, low-cost hardware suffers from defects like power amplifier nonlinearity, oscillator phase noise, and quantization errors due to the use of low-resolution DAC. Moreover, due to the finite resolution of the physical components, an RIS is affected by phase noise, the distribution of which is usually characterized by a uniform distribution or a von Mises distribution [24]--[26]. Therefore, if the hardware impairments (HWI) are not carefully considered, severe performance degradations may occur [27]. Currently, there have been several studies on improving the secrecy performance in RIS-assisted MIMO systems under HWI. For example, an RIS-aided secure multiple-input single-output (MISO) system was considered in the presence of HWI [28]. The authors of [29] introduced a fairness-based scheme to optimize the weighted minimum secrecy rate in a hardware-impaired RIS-aided multiuser MISO downlink channel. In [30], the impact of phase-shift errors on the secrecy rate was inspected in an RIS-aided MISO uplink channel, where the required user power was deduced under a specified target secrecy rate.

However, the studies reported in [14], [16], [18]--[22] assume perfect CSI at the BS, while obtaining CSI in RIS-aided systems is a difficult task in contrast to conventional systems without an RIS. In this regard, an uplink channel estimation method was introduced in [23] in the absence of prior information about active pilot attacks. Only perfect transceivers and RIS hardware were considered. The work in [31] examined the impact of imperfect CSI on the secrecy rate in an RIS-aided non-orthogonal multiple access (NOMA) system, but no comprehensive analysis on the secrecy rate was provided. Furthermore, because of the physical finite size of devices and the small interdistances among the RIS elements and BS antennas, the spatial correlation cannot be ignored [32]--[35]. In that regard, spatially correlated fading channels have been taken into account considering a multiple-antenna Eve in [22] and a single-antenna Eve in [23], respectively. However, the analyses reported in [22] and [23] are not immediately applicable to RIS-aided secure networks due to the presence of both CSI and hardware imperfections. More importantly, the legitimate users may employ low-cost hardware components that also contribute to the HWI, while Eve is expected to use high-quality equipment with negligible HWI. Therefore, it is indispensable to understand theoretical performance bounds in terms of communication secrecy with the aid of an RIS.

Motivated by these observations, we analyze the secrecy performance of an RIS-assisted multiuser massive MIMO system, in which the impact of CSI imperfection, RIS phase noise, spatial correlation, and transceiver HWI are quantitatively characterized. The main contributions of this paper are summarized as follows:
\begin{itemize}
\item We derive new closed-form expressions for the achievable ergodic secrecy rate considering AN and realistic constraints including RIS phase noise, imperfect CSI, and transceiver HWI. To the best of our knowledge, the analytical characterization of the secrecy performance in RIS-aided massive MIMO systems under these practical operating conditions is still not available in the literature.
\item From the analytical expressions of the ergodic secrecy rate, several insights on the impact of HWI and imperfect CSI on the secrecy performance are obtained. In particular, we show that the ergodic secrecy rate increases logarithmically with the number of BS antennas while it approaches a finite limit as the number of RIS elements increases without bounds. As a major outcome, we prove that non-zero secrecy rates are sustained when the transmit power scales as $1/N$ with the number of RIS elements $N$ .
\item We develop a power allocation strategy to alleviate the impact of HWI and deduce a closed-form solution for the optimal fraction of power to be allocated between the information data and AN. Based on the obtained theoretical and numerical results, we show that there are situations in which the transmit distortion noise improves the secrecy rate by having an impact akin to the AN. Furthermore, we demonstrate that increasing the number of RIS elements may counteract the reduction of secrecy performance imposed by the transceiver HWI and RIS phase distortion. Nonetheless, the secrecy rate may be compromised by the presence of spatial correlation among the RIS elements.
\end{itemize}

The remainder of this paper is organized as follows. The system model is introduced in Section II. In Section III, the uplink training and downlink transmission strategies under HWI are described. A tractable lower bound for the secrecy rate and simple analytical expressions in some asymptotic regimes are presented in Section IV. In Section V, numerical results are illustrated. Finally, Section VI concludes the paper.

\emph{Notation}: ${\bf X}^{-1}$ and ${\bf X}^H$ denote the inverse and conjugate transpose of the matrix $\bf X$, respectively. $\mathbb E\{\cdot\}$ and $\mathbb V\rm ar\{\cdot\}$ denote the expectation and variance of a random variable, respectively. The complex Gaussian distribution with zero mean and $\sigma^2$ variance is denoted by $\mathcal {CN}(0,\sigma^2)$. The trace of the matrix $\bf X$ is denoted by ${\rm tr}({\bf X})$. Besides, $[{\bf X}]_{a,b}$ denotes the $(a,b)$-th element of the matrix. $\mathbb C^{a\times b}$ is the set of $a\times b$ complex matrices, $\circ$ represents the Hadamard product, and $\lfloor \cdot \rfloor$ is the floor function.

\section{System Model}
\begin{figure}[t]
	\centering
	\centerline{\includegraphics[width=0.4\textwidth]{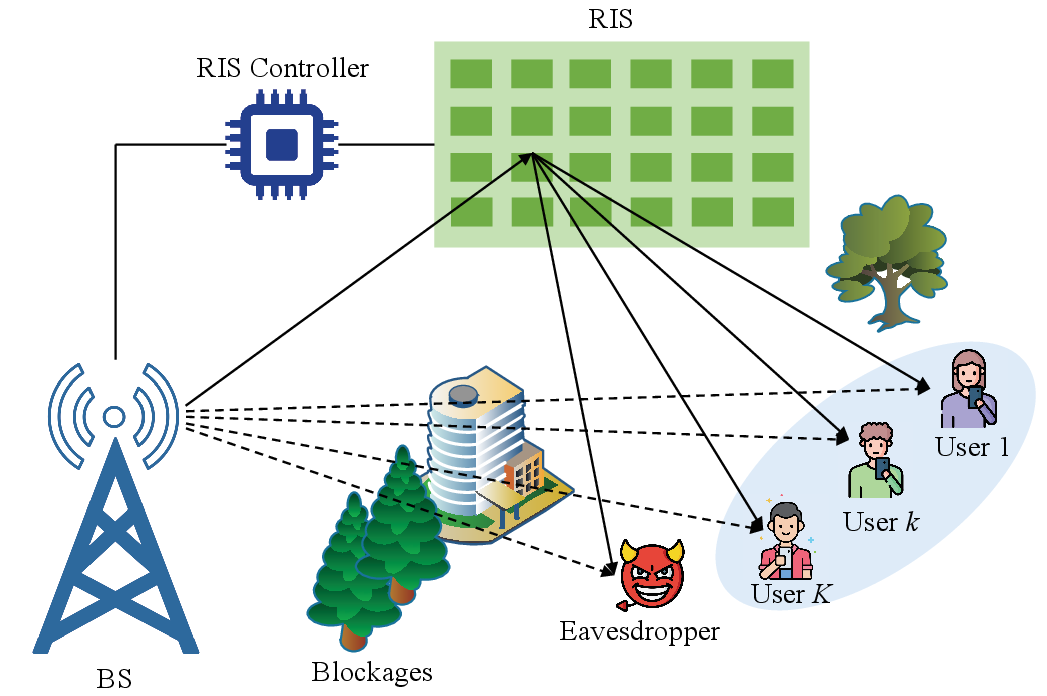}}
	\caption{RIS-aided downlink secure massive MIMO system in the presence of an $M_{\rm E}$-antenna Eve}
\end{figure}
We consider the downlink of an RIS-aided multiuser secure massive MIMO system as shown in Fig. 1, where a BS equipped with $M$ antennas communicates with $K$ single-antenna legitimate users with the assistance of an RIS equipped with $N$ reflecting elements. One $M_{\rm E}$-antenna Eve is located close to the legitimate users. This assumption can be applied to the scenario in which multiple Eves cooperate to wiretap the same confidential information. The RIS receives processing messages through a smart controller under the coordination of the BS. We assume that the CSI of Eve's channel is accessible at the BS, which is a widely used case study in literature, e.g., [12], [22], [36]. Also, we assume that the system operates under a standard time division duplex (TDD) protocol, where the uplink channel estimation and downlink data transmission are carried out during each coherence time block.

The channels of the BS-RIS link, RIS-user link, BS-user link, RIS-Eve link, and BS-Eve link are assumed to be flat-fading channels, and are denoted by ${\bf H}_1\in\mathbb C^{M\times N}$, ${\bf h}_{{\rm I},k}\in\mathbb C^{N\times 1}$, ${\bf h}_{{\rm B},k}\in\mathbb C^{M\times 1}$, ${\bf H}_{\rm I,E}\in\mathbb C^{N\times M_{\rm E}}$, and ${\bf H}_{\rm B,E}\in\mathbb C^{M\times M_{\rm E}}$, respectively. Specifically, we consider spatially correlated Rayleigh fading channels given by
\begin{align}
	{\bf h}_{{\rm I},k}={\bf R}_{{\rm I},k}^{1/2}{\bf g}_{{\rm I},k}\quad {\text {and}}\quad {\bf h}_{{\rm B},k}={\bf R}_{{\rm B},k}^{1/2}{\bf g}_{{\rm B},k},
\end{align}
where ${\bf R}_{{\rm I},k}$ and ${\bf R}_{{\rm B},k}$ denote the spatial correlation matrices at the RIS and BS, respectively, and ${\bf g}_{{\rm I},k}\sim\mathcal{CN}(0,{\bf I}_N)$ and ${\bf g}_{{\rm B},k} \sim\mathcal{CN}(0,{\bf I}_M)$ are the small-scale fading vectors. Similar to the channel model of the legitimate users, the channels of the links from the RIS/BS to Eve are given by
\begin{align}
	{\bf H}_{\rm I,E}={\bf R}_{\rm I,E}^{1/2}{\bf G}_{\rm I,E}\quad {\text {and}}\quad  {\bf H}_{\rm B,E}={\bf R}_{\rm B,E}^{1/2}{\bf G}_{\rm B,E},
\end{align}
where ${\bf R}_{\rm I,E}$ and ${\bf R}_{\rm B,E}$ denote the spatial correlation matrices and the elements of ${\bf G}_{\rm I,E}$ and ${\bf G}_{\rm B,E}$ follow the distribution $\mathcal{CN}(0,1)$. The large-scale fading coefficients are absorbed into the spatial correlation matrices, that is, ${\bf R}_{{\rm B},k}$ and ${\bf R}_{\rm B,E}$ are defined as ${\bf R}_{{\rm B},k}=\beta_{2,k}{\bf R}_{\rm B}$ and ${\bf R}_{\rm B,E}=\beta_3{\bf R}_{\rm B}$, where $\beta_{2,k}$ and $\beta_3$ are the path losses. Moreover, ${\bf R}_{{\rm I},k}$ and ${\bf R}_{\rm I,E}$ can be expressed as ${\bf R}_{{\rm I},k}=\beta_{{\rm I},k}{\bf R}_{\rm I}$ and ${\bf R}_{\rm I,E}=\beta_{\rm I,E}{\bf R}_{\rm I}$, where $\beta_{{\rm I},k}$ and  $\beta_{\rm I,E}$ account for the large-scale path losses of the channels ${\bf h}_{{\rm I},k}$ and ${\bf H}_{\rm I,E}$, respectively. The $(x,y)$-th element of the matrix ${\bf R}_{\rm I}$ can be written as [32]
\begin{align}
	[{\bf R}_{\rm I}]_{x,y}={\rm sinc}\left(\frac{2\Vert{\bf c}_x-{\bf c}_y\Vert}{\lambda}\right),
\end{align}
where ${\bf c}_x=[0,{\rm mod}(x-1,N_{\rm H})d_h,\lfloor(x-1)/N_{\rm V}\rfloor d_v]^T$, $\lambda$ is the wavelength, and $N_{\rm H}$ and $N_{\rm V}$ represent the number of RIS elements of the UPA on the horizontal and vertical directions, respectively. ${\bf R}_{\rm B}$ and ${\bf R}_{\rm I}$ are nonnegative deterministic matrices having traces ${\rm tr}({\bf R}_{\rm B})=M$ and ${\rm tr}({\bf R}_{\rm I})=N$.

Since the BS and the RIS are usually located at high elevations and in predetermined locations, the channel between the BS and the RIS is assumed to be in line of sight (LoS) and hence we ignore any non-LoS (NLoS) paths for this channel [26]. LoS channels are usually obtained when the RIS is placed close to the BS [37]. Specifically, the LoS channel ${\bf H}_1$ is modeled as $[{\bf H}_1]_{m,n}=\sqrt{\beta_1}\bar{\bf H}_1$, where $\beta_1$ is the large-scale path loss and $\bar{\bf H}_1$ is formulated similar to [22]. Furthermore, the RIS phase shift matrix is denoted by ${\bf \Phi}={\rm diag}\left(e^{j\theta_1}, \dots ,e^{j\theta_N}\right)\in\mathbb C^{N\times N}$ with $\theta_n\in[0, 2\pi), \forall n=1, \dots ,N$. In most works, the RIS phase shift is assumed to be perfectly designed [28], [31]. Owing to the presence of HWI, however, a phase-shift error usually exists, and it is modeled as ${\bf \Theta}={\rm diag}\left(e^{j\tilde{\theta}_1}, \dots ,e^{j\tilde{\theta}_N}\right)$, where $\tilde{\theta}_n,\forall n=1, \dots ,N$ are often assumed to be von Mises or uniform distributed random variables with zero mean, i.e., $\tilde{\theta}_n\sim\mathcal {VM}(0,\nu_p)$ and $\mathcal U[-\iota_p,\iota_p]$, respectively [38]. Here, $\nu_p$ is the concentration parameter, and we denote the phase noise power by $\sigma_p^2$, corresponding to $\nu_p=1/\sigma_p^2$ and $\iota_p=\sqrt{3\sigma_p^2}$.

\section{Channel Estimation and Data Transmission}
In this section, we first use the linear minimum mean square error (LMMSE) method to estimate the aggregate instantaneous channels. Then, we present the downlink data transmission scheme by introducing the impact of imperfect hardware at both the RIS and the transceiver.

\subsection{Uplink Channel Estimation}
Here, we present the uplink channel estimation assuming that the RIS phase shifts are given. Considering a TDD system, the CSI can be obtained by uplink training assuming channel reciprocity [23], [26], [39]. We assume that the legitimate user~$k$ transmits a pilot sequence of ${\boldsymbol\phi}_k\in\mathbb C^{\tau_u\times 1}$ in the uplink training interval $\tau_u$, where ${\boldsymbol\phi}_k^H{\boldsymbol\phi}_k=\tau_u$ and $\vert\phi_{k,i}\vert^2=1,\forall k,i$. The pilot signals are designed to be orthogonal among the users, i.e., ${\boldsymbol\phi}_k^H{\boldsymbol\phi}_i=0,\forall i\neq k$. We assume $\tau_u\geq K$ to ensure that there is no pilot contamination. As such, the received hardware-impaired pilot matrix ${\bf Y}_p\in\mathbb C^{M\times \tau_u}$ is given by
\begin{align}
{\bf Y}_p=\sum\limits_{i=1}^{K}({\bf H}_1{\bf \Phi}{\bf \Theta}{\bf h}_{{\rm I},i}+{\bf h}_{{\rm B},i}
)\left(\sqrt{\rho }{\boldsymbol\phi}_i^H+{\boldsymbol \eta}_{t,i}^H\right)+{\boldsymbol \Upsilon}_r^{\rm BS}+{\bf N}_p,
\end{align}
where ${\boldsymbol \eta}_{t,i}\in\mathbb C^{\tau_u\times 1}$ and ${\boldsymbol \Upsilon}_r^{\rm BS}\in\mathbb C^{M\times \tau_u}$ denote the additive RF hardware distortion of the $i$-th legitimate user and the BS, respectively. ${\bf N}_p\in\mathbb C^{M\times\tau_u}$ is the additive white Gaussian noise (AWGN) matrix with entries distributed as $\mathcal {CN}(0,\sigma_u^2)$. Based on (4), we denote the aggregate channel of the legitimate user $k$ as ${\bf h}_k = {\bf H}_1{\bf \Phi}{\bf \Theta}{\bf h}_{{\rm I},k}+{\bf h}_{{\rm B},k}$. In particular, the distribution of the entries of ${\boldsymbol \eta}_{t,k}$ is $\mathcal {CN}(0,\rho\kappa_t^{\rm UE}{\bf C}_t)$ with ${\bf C}_t={\rm diag}\left(\mathbb E\{{\boldsymbol\phi}_k{\boldsymbol\phi}_k^H\}\right)$, while each column of ${\boldsymbol \Upsilon}_r^{\rm BS}$ is distributed as ${\boldsymbol \upsilon}_r^{\rm BS}\sim\mathcal {CN}(0, \rho\kappa_r^{\rm BS}{\bf D}_r)$ with ${\bf D}_r=\sum_{i=1}^{K}\mathbb E\{{\boldsymbol\phi}_k{\boldsymbol\phi}_k^H\} {\rm diag}\big(\vert h_{i,1}\vert^2,\dots ,\vert h_{i,m}\vert^2, \dots ,\vert h_{i,M}\vert^2\big)$ and $h_{i,m}$ being the $m$-th element of ${\bf h}_i\ (m=1, \dots ,M)$. The variance is equal to ${\bf D}_r=\sum_{i=1}^{K}\mathbb E\{{\boldsymbol\phi}_k{\boldsymbol\phi}_k^H\}{\bf I}_M\circ {\bf h}_i{\bf h}_i^H$ [40].

To estimate ${\bf h}_k$, the received signal ${\bf Y}_p$ is multiplied by the pilot signal ${\boldsymbol\phi}_k$ to eliminate the interference. Then, we have
\begin{align}
{\bf y}_{p,k}=\ &{\bf Y}_p{\boldsymbol\phi}_k\nonumber\\
=\ &\tau_u\sqrt{\rho}{\bf h}_k+\sum\limits_{i=1}^{K}{\boldsymbol \eta}_{t,i}^H{\boldsymbol\phi}_k{\bf h}_i+{\boldsymbol \Upsilon}_r^{\rm BS}{\boldsymbol\phi}_k+{\bf N}_p{\boldsymbol\phi}_k.
\end{align}
The BS uses the signal ${\bf y}_k^p$ to obtain the LMMSE estimate of ${\bf h}_k$. This channel estimator is derived as follows.
\begin{lemma}
The LMMSE estimate of ${\bf h}_k$ is obtained as
\begin{align}
\hat{\bf h}_k=\sqrt{\rho}{\bf R}_k{\bf\Psi}_k^{-1}{\bf y}_{p,k},
\end{align}
where
\begin{align}
&{\bf\Psi}_k=\tau_u\rho{\bf R}_k+\rho\kappa_t^{\rm UE}\sum_{i=1}^{K}{\bf R}_i+\rho\kappa_r^{\rm BS}\sum_{i=1}^{K}{\bf I}_M\circ{\bf R}_i+\sigma_u^2{\bf I}_M,\\
&{\bf R}_k={\bf R}_{{\rm B},k}+{\bf H}_1{\bf \Phi}\widetilde{\bf R}_{{\rm I},k}{\bf \Phi}^H{\bf H}_1^H,
\end{align}
and
\begin{align}
\widetilde{\bf R}_{{\rm I},k}=\varrho^2{\bf R}_{{\rm I},k}+\beta_{{\rm I},k}(1-\varrho^2){\bf I}_N,
\end{align}
with $\varrho=\frac{I_1(\nu_p)}{I_0(\nu_p)}$ for the von Mises distribution and $\varrho=\frac{\sin(\iota_p)}{\iota_p}$ for the uniform distribution.
\end{lemma}
\begin{proof}
See Appendix A.
\end{proof}

Since only the aggregated channel with dimension $M\times K$ is estimated in the uplink, the training overhead for the LMMSE scheme is $K$. According to the channel estimate in (6), the normalized mean square error (NMSE) of ${\bf h}_k$ is given by
\begin{align}
{\rm NMSE}_k=\frac{\mathbb E\big\{\Vert{\bf h}_k-\hat{\bf h}_k\Vert^2\big\}}{\mathbb E\big\{\Vert{\bf h}_k\Vert^2\big\}}=\frac{{\rm tr}({\bf C}_k)}{{\rm tr}({\bf R}_k)},
\end{align}
where ${\bf C}_k={\bf R}_k-\tau_u\rho{\bf R}_k{\bf \Psi}_k^{-1}{\bf R}_k$. We note that $\hat{\bf h}_k$ and the estimation error ${\bf e}_k={\bf h}_k-\hat{\bf h}_k$ are uncorrelated but not necessarily independent, with $\mathbb E\{\hat{\bf h}_k\hat{\bf h}_k^H\}=\tau_u\rho{\bf R}_k{\bf \Psi}_k^{-1}{\bf R}_k$ and $\mathbb E\{{\bf e}_k{\bf e}_k^H\}={\bf C}_k$. We now examine the asymptotic NMSE under the assumption of a high pilot power.
\begin{corollary}
Using (6) and (10), the ${\rm NMSE}_k$ of the LMMSE estimator for $\rho \rightarrow\infty$ is given by
\begin{align}
{\rm NMSE}_k=\frac{{\rm tr}({\bf R}_k-\tau_u{\bf R}_k\widetilde{\bf \Psi}_k^{-1}{\bf R}_k)}{{\rm tr}({\bf R}_k)},
\end{align}
where $\widetilde{\bf\Psi}_k=\tau_u{\bf R}_k+\kappa_t^{\rm UE}\sum_{i=1}^{K}{\bf R}_i+\kappa_r^{\rm BS}\sum_{i=1}^{K}{\bf I}_M\circ{\bf R}_i$.
\end{corollary}

It is inspected from (11) that, as $\rho \rightarrow\infty$, the NMSE$_k$ goes to zero if the hardware is assumed to be ideal, i.e., $\kappa_t^{\rm UE} = \kappa_r^{\rm BS} = 0$. We further give asymptotic results to analyze how the number of RIS elements affects the accuracy of the estimation.

\begin{corollary}
When $N\gg M$, ${\bf R}_{{\rm B}}={\bf R}_{{\rm I}}={\bf I}_N$, and $\kappa_t^{\rm UE} = \kappa_r^{\rm BS}=\varrho= 0$, the ${\rm NMSE}_k$ is given by
\begin{align}
{\rm NMSE}_k\rightarrow1-\frac{\beta_{2,k}+\beta_{{\rm I},k}\beta_1 N}{\beta_{2,k}+\beta_{{\rm I},k}\beta_1 N+\frac{\sigma_u^2}{\tau_u\rho}}.
\end{align}
\end{corollary}
\begin{proof}
	We simplify the corresponding terms in (10) as
	\begin{align}
	&{\rm NMSE}_k=1-\nonumber\\
	&\frac{\tau_u\rho{\rm tr}\left(\left(\beta_{2,k}{\bf I}_M+\beta_{{\rm I},k}{\bf H}_1{\bf H}_1^H\right)\ddot{\bf \Psi}_k^{-1}\left(\beta_{2,k}{\bf I}_M+\beta_{{\rm I},k}{\bf H}_1{\bf H}_1^H\right)\right)}{{\rm tr}\left(\beta_{2,k}{\bf I}_M+\beta_{{\rm I},k}{\bf H}_1{\bf H}_1^H\right)},
	\end{align}
	where $\ddot{\bf \Psi}_k=\tau_u\rho\left(\beta_{2,k}{\bf I}_M+\beta_{{\rm I},k}{\bf H}_1{\bf H}_1^H\right)+\sigma_u^2{\bf I}_M$. Furthermore, we have the property ${\bf H}_1{\bf H}_1^H\rightarrow\beta_1N{\bf I}_M$ when $N\gg M$. By keeping the dominant terms in (13), we obtain the asymptotic expression in (12).
\end{proof}

Corollary~2 shows that NMSE$_k$ is a decreasing function of $N$. When $N\rightarrow\infty$, NMSE$_k$ approaches zero because the RIS boosts the channel gain, leading to a smaller normalized error. On the other hand, the MSE matrix ${\bf C}_k$ reduces to $\frac{1}{\frac{1}{\beta_{2,k}+\beta_{{\rm I},k}\beta_1 N}+\frac{\tau_u\rho}{\sigma_u^2}}{\bf I}_M$ and converges to $\frac{\sigma_u^2}{\tau_u\rho}{\bf I}_M$ as $N\rightarrow\infty$.

\subsection{Downlink Data Transmission}
Since the AN serves as an effective method to enhance the secrecy performance, we consider transmitting the information signal concurrently with the AN at the BS. We denote ${\bf s}=[s_1, \dots ,s_k, \dots ,s_K]^T$ as the information symbols, satisfying $\mathbb E\{{\bf s}{\bf s}^H\}={\bf I}_K$, and let ${\bf z}\sim\mathcal{CN}({\bf 0},{\bf I}_{M-K})$ represent the AN. Then, the information signal and the AN are multiplied by the precoding matrices ${\bf W}\in\mathbb C^{M\times K}$ and ${\bf V}\in\mathbb C^{M\times (M-K)}$, respectively. Accordingly, the transmitted signal ${\bf x}\in\mathbb C^{M\times 1}$ is written as
\begin{align}
	{\bf x}=\ &\sqrt{\frac{\xi P_t}{K}}{\bf Ws}+\sqrt{\frac{(1-\xi) P_t}{M-K}}{\bf Vz},
\end{align}
where $\xi\in(0,1]$ is the power allocation coefficient between ${\bf s}$ and ${\bf z}$, and $P_t$ is the total transmit power. Note that the power of the precoding matrices are constrained by ${\rm tr}({\bf W}{\bf W}^H)=K$ and ${\rm tr}({\bf V}{\bf V}^H)=M-K$. For ease of notation, in the following we define $p\triangleq\frac{\xi P_t}{K}$ and $q\triangleq\frac{(1-\xi) P_t}{M-K}$ .

Then, the received signal at the legitimate user $k$ is given by
\begin{equation}
	y_k={\bf h}_k^H({\bf x}+{\boldsymbol \eta}_t^{\rm BS})+\eta_{r,k}^{\rm UE}+n_k,
\end{equation}
where $n_k\sim\mathcal{CN}(0,\sigma_k^2)$ is the AWGN at the legitimate user~$k$. Additionally, ${\boldsymbol \eta}_t^{\rm BS}\sim\mathcal {CN}({\bf 0}, {\bf \Upsilon}_t^{\rm BS})$ and $\eta_{r,k}^{\rm UE}\sim\mathcal {CN}(0,\mu_{r,k})$ denote the downlink transmitter HWI and the receiver HWI of the legitimate user $k$, respectively, where
\begin{align}
&{\bf \Upsilon}_t^{\rm BS}=\kappa_t^{\rm BS}{\rm diag}\left(T_{11}, \dots ,T_{ii}, \dots ,T_{MM}\right)\ \ \ \ {\text {and}} \nonumber\\
&\mu_{r,k}^{\rm UE}=\kappa_r^{\rm UE}{\bf h}_k^H{\bf T}{\bf h}_k
\end{align}
with ${\bf T}=\mathbb E\{{\bf xx}^H\}$ and $T_{ii}=[{\bf T}]_{ii},i=1, \dots ,M$ [36]. Here, we assume that the parameters $\kappa_t^{\rm BS}$ and $\kappa_r^{\rm UE}$ remain the same for all the users. Combining (14) and (15), the signal received at the legitimate user $k$ is written as
\begin{align}
y_k=\ &\sqrt{p}{\bf h}_k^H{\bf w}_k s_k+\sqrt{p}\sum\nolimits_{i\neq k}^{K}{\bf h}_k^H{\bf w}_i s_i+\sqrt{q}{\bf h}_k^H{\bf Vz}\nonumber\\
&+{\bf h}_k^H{\boldsymbol \eta}_t^{\rm BS}+\eta_{r,k}^{\rm UE}+n_k,
\end{align}
where ${\bf w}_k$ denotes the $k$-th column of the precoding matrix~${\bf W}$.

As the BS does not have access to Eve's capability, we consider the best-case scenario for Eve in terms of useful signal leakage: i.e., high-quality hardware is employed. Therefore, the received signal contains only distortion noise generated by the HWI at the transmitter. In particular, the signal received at Eve is modeled as
\begin{align}
	{\bf y}_{\rm E}=\sqrt{p}{\bf H}_{\rm E}^H{\bf Ws}+\sqrt{q}{\bf H}_{\rm E}^H{\bf Vz}+{\bf H}_{\rm E}^H{\boldsymbol \eta}_t^{\rm BS}+{\bf n}_{\rm E},
\end{align}
where ${\bf H}_{\rm E}={\bf H}_{\rm B,E}+{\bf H}_1{\bf \Phi}{\bf\Theta}{\bf H}_{\rm I,E}$ is the aggregate channel between the BS and Eve and ${\bf n}_{\rm E}\sim\mathcal{CN}({\bf 0},\sigma_{\rm E}^2{\bf I}_{M_{\rm E}})$ is the AWGN at Eve. As a benchmark, we assume that Eve is fully capable of estimating its effective channel and can decode and eliminate the interference from other legitimate users [13], [36]. Hence, an upper bound for the ergodic capacity of Eve can be expressed as [41]
\begin{equation}
	C_{\rm E}=\mathbb E\big\{\log_2(1+\gamma_{\rm E})\big\},
\end{equation}
where ${\gamma_{\rm E}}$ is defined as
\begin{align}
{\gamma_{\rm E}}\triangleq p({\bf f}_{k,{\rm E}})^H\left({\bf H}_{\rm E}^H\left(q{\bf V}{\bf V}^H+{\bf \Upsilon}_t^{\rm BS}\right){\bf H}_{\rm E}\right)^{-1}{\bf f}_{k,{\rm E}},
\end{align}
where ${\bf f}_{k,{\rm E}}={\bf H}_{\rm E}^H{\bf w}_k$.  Considering the passive nature of Eve, it is assumed that the thermal noise at Eve is negligible, i.e., $\sigma_{\rm E}^2\rightarrow 0$.

\section{Achievable Ergodic Secrecy Rate Analysis}
In this section, we utilize the channel statistics of the legitimate users and Eve to determine the achievable ergodic secrecy rate in the presence of HWI. Therefore, the obtained expression depends only on long-term channel coefficients, thus it can be used to optimize the power allocation strategy with low overhead. To this end, we first provide a closed-form formula for the achievable rate of the legitimate users when the AN is employed at the BS. Then, we present an upper bound for the ergodic capacity of Eve. The obtained analytical results are exploited to gain valuable insights on the impact of common design parameters.

\subsection{Achievable Rate of the Legitimate Users}
In order to obtain a clear picture of the role of an RIS for enhancing the secrecy performance, the analysis is conducted by considering a low-complexity maximum ratio transmission (MRT) scheme. Specifically, the MRT is given by ${\bf w}_k =\frac{\hat{\bf h}_k}{\sqrt{\mathbb E\{\Vert \hat{\bf h}_k \Vert^2\}}}$ with $\mathbb E\{\Vert {\bf w}_k\Vert^2\}=1$ and the precoding matrix is expressed as
\begin{align}
{\bf W}=[{\bf w}_1, \dots ,{\bf w}_k, \dots ,{\bf w}_K].
\end{align}
As for the AN signal, the objective is to design ${\bf V}=[{\bf v}_1, \dots , {\bf v}_i, \dots, {\bf v}_{M-K}]$ with $\Vert{\bf v}_i\Vert^2=1$ such that it lies in the null space of $\hat{\bf h}_k$, i.e., $\hat{\bf h}_k^H{\bf V} ={\bf 0}_{M-K}, k=1, \dots ,K$\footnote{From the standpoint of theoretical analysis, $M>K$ is considered to facilitate the design of the AN.}.

For massive MIMO, like in [26] and [37], we take advantage of the channel hardening property. Accordingly, the effective channel of the legitimate users converges to its average value. Then, $y_k$ in (17) is rewritten as
\begin{align}
y_k=\sqrt{p}\mathbb E\{{\bf h}_k^H{\bf w}_k\} s_k+z_k,
\end{align}
where $z_k$ is expressed as
\begin{align}
z_k=\ &\sqrt{p}\left({\bf h}_k^H{\bf w}_k-\mathbb E\{{\bf h}_k^H{\bf w}_k\}\right)s_k+\sqrt{p}\sum\nolimits_{i\neq k}^K{\bf h}_k^H{\bf w}_i s_i\nonumber\\
&+\sqrt{q}{\bf h}_k^H{\bf Vz}+{\bf h}_k^H{\boldsymbol \eta}_t^{\rm BS}+\eta_{r,k}^{\rm UE}+n_k.
\end{align}
Based on (22) and (23), the achievable rate of the legitimate user $k$ is
\begin{equation}
R_k=\log_2(1+\gamma_k),
\end{equation}
\begin{figure*}[ht]
	\begin{equation}
	\gamma_k=\frac{p\left\vert\mathbb E\left\{{\bf h}_k^H{\bf w}_k\right\}\right\vert^2}{p\sum\nolimits_{i\neq k}^K\mathbb E\left\{\left\vert{\bf h}_k^H{\bf w}_i\right\vert^2\right\}+p\mathbb V{\rm ar}\left\{{\bf h}_k^H{\bf w}_k\right\}+q\mathbb E\left\{{\bf h}_k^H{\bf V}{\bf V}^H{\bf h}_k\right\}+\mathbb E\left\{{\bf h}_k^H{\bf\Upsilon}_t^{\rm BS}{\bf h}_k\right\}+\mathbb E\{\mu_{r,k}\}+\sigma_k^2}.
	\end{equation}
		\hrulefill
\end{figure*}where the received SNR $\gamma_k$ is written as (25), shown at the top of the next page. Based on (24), we provide a closed-form expression for the achievable rate of the legitimate user $k$ in the following theorem.

\begin{theorem}
	 The achievable downlink rate of the legitimate user $k$ in the presence of HWI is expressed as follows:
	\begin{equation}
		R_k=\log_2\left(1+S_k/I_k\right),
	\end{equation}
	where
	\begin{align}
		&S_k=p\tau_u\rho{\rm tr}({\bf R}_k{\bf \Psi}_k^{-1}{\bf R}_k),\\
		&I_k=p\sum\nolimits_{i\neq k}^K\frac{{\rm tr}({\bf R}_k{\bf R}_i{\bf \Psi}_i^{-1}{\bf R}_i)}{{\rm tr}({\bf R}_i{\bf \Psi}_i^{-1}{\bf R}_i)}+p\frac{{\rm tr}\left({\bf C}_k{\bf R}_k{\bf \Psi}_k^{-1}{\bf R}_k\right)}{{\rm tr}\left({\bf R}_k{\bf \Psi}_k^{-1}{\bf R}_k\right)}\nonumber\\
		&+q\frac{M-K}{M}{\rm tr}\left({\bf C}_k\right)+\left(\kappa_t^{\rm BS}+\kappa_r^{\rm UE}\right)\frac{P_t}{M}{\rm tr}({\bf R}_k)+\sigma_k^2.
	\end{align}
\end{theorem}

\begin{proof}
	See Appendix B.
\end{proof}

It can be seen that an intricate expression combined with multiple correlation matrices determines the achievable rate of the legitimate user $k$, which relies on the statistical parameters of the legitimate users. We observe that the signal power increases linearly with $M$ while the interference power $\frac{{\rm tr}({\bf R}_k{\bf R}_i{\bf \Psi}_i^{-1}{\bf R}_i)}{{\rm tr}({\bf R}_i{\bf \Psi}_i^{-1}{\bf R}_i)}$ remains roughly constant with varying $M$. The expression in (28) includes the multiuser interference, AN, and HWI that degrade the achievable rate. Besides, the interference power is evaluated by the similarity of the spatial correlation matrices ${\bf R}_k$ and ${\bf R}_i$: It decreases when the spatial correlation among legitimate users becomes more distinct, which is the so-called ``favorable propagation conditions'' [42].

\subsection{Upper Bounded Capacity at Eve}
Considering the worst-case situation discussed in Section~III, the following theorem yields a tractable tight upper bound for (19).

\begin{theorem}
	 For $M\rightarrow\infty$ and assuming the null-space AN precoding, the ergodic capacity of Eve in (19) is upper bounded as
	\begin{equation}
	C_{\rm E}\leq\overline C_{\rm E}=\log_2\left(1+S_{\rm E}/I_{\rm E}\right),
	\end{equation}
	where
	\begin{align}
		S_{\rm E}=\ &pM_{\rm E}M\left[q(M-K)+\kappa_t^{\rm BS}P_t\right]\nonumber\\
		\ \ \ &\times{\rm tr}\left({\bf R}_k{\bf \Psi}_k^{-1}{\bf R}_k{\bf Q}_{\rm E}\right){\rm tr}({\bf Q}_{\rm E}),\\
		I_{\rm E}=\ &\chi{\rm tr}({\bf R}_k{\bf \Psi}_k^{-1}{\bf R}_k),
	\end{align}
	with the parameter
	\begin{align}
		\chi=\ &\left(q(M-K)+\kappa_t^{\rm BS}P_t\right)^2\left[{\rm tr}({\bf Q}_{\rm E})\right]^2-M_{\rm E}\left[\left(\kappa_t^{\rm BS}P_t\right)^2   \right.\nonumber\\
		&\left. +q^2 M(M-K) +2q (M-K)\kappa_t^{\rm BS}P_t\right]
		{\rm tr}\left({\bf Q}_{\rm E}^2\right).
	\end{align}
	Here, ${\bf Q}_{\rm E}={\bf R}_{\rm B,E}+{\bf H}_1{\bf \Phi}\widetilde{\bf R}_{\rm I, E}{\bf \Phi}^H{\bf H}_1^H$ represents the covariance matrix of the aggregate channel of Eve with $\widetilde{\bf R}_{\rm I, E}=\varrho^2{\bf R}_{\rm I,E}+\beta_{\rm I,E}(1-\varrho^2){\bf I}_N$.
\end{theorem}
\begin{proof}
See Appendix C.
\end{proof}

We observe from (29) that the capacity of Eve, as expected, increases with the number of $M_{\rm E}$. If no AN is utilized at the transmitter, the upper bound for $C_{\rm E}$ in (29) simplifies to
\begin{equation}
\overline C_{\rm E}\big\vert_{q=0}=\log_2\left(1+\frac{M_{\rm E} M{\rm tr}\left({\bf R}_k{\bf \Psi}_k^{-1}{\bf R}_k{\bf Q}_{\rm E}\right){\rm tr}({\bf Q}_{\rm E})}{\kappa_t^{\rm BS}\zeta K \big(\left[{\rm tr}({\bf Q}_{\rm E})\right]^2-M_{\rm E}{\rm tr}\left({\bf Q}_{\rm E}^2\right)\big)}\right),
\end{equation}
where $\zeta = {\rm tr}({\bf R}_k{\bf \Psi}_k^{-1}{\bf R}_k)$. Specifically, $\overline C_{\rm E}\rightarrow\infty$ holds true when $\kappa_t^{\rm BS}\rightarrow 0$, which implies that secure communications are not possible.

\subsection{Ergodic Secrecy Rate Analysis}
In this section, we investigate theoretical lower bounds for the ergodic secrecy rate. By virtue of [36], the ergodic secrecy rate is given by
\begin{equation}
	R_{\rm sec}=[R_k-\overline C_{\rm E}]^+,
\end{equation}
where $[a]^+={\rm max}\{0, a\}$. The ergodic secrecy rate may be readily calculated in closed form by substituting (26) and (29) into (34). We notice that the ergodic secrecy rate is related to the RIS phase shift matrix in a rather complex way. However, the derived expression can be utilized for optimizing the RIS phase shifts, by only relying on long-term channel statistics. As the ergodic secrecy rate depends on the correlation matrices, it is difficult to quantitatively inspect their impact on the ergodic secrecy rate. This will be examined and discussed in the simulation results.

For the special case without the AN, the ergodic secrecy rate of the legitimate user $k$ simplifies to (35), shown at the top of the next page, where
\begin{figure*}[ht]
\begin{equation}
R_{\rm sec}\big\vert_{q=0}=\left[\log_2\left(1+\frac{P_t \ddot S_k /K }{P_t\ddot I_k /K+\ddot N_k}\right)-\log_2\left(1+\frac{\delta M^2{\rm tr}\left({\bf R}_k{\bf \Psi}_k^{-1}{\bf R}_k{\bf Q}_{\rm E}\right){\rm tr}({\bf Q}_{\rm E})}{\kappa_t^{\rm BS}\zeta K\left[{\rm tr}({\bf Q}_{\rm E})\right]^2-\delta \kappa_t^{\rm BS}\zeta MK{\rm tr}\left({\bf Q}_{\rm E}^2\right)}\right)\right]^+,
\end{equation}
\hrulefill
\end{figure*}
\begin{align}
\ddot S_k =\ &\tau_u\rho{\rm tr}({\bf R}_k{\bf \Psi}_k^{-1}{\bf R}_k),\\
\ddot{I_k}=\ &\sum\nolimits_{i\neq k}^K\frac{{\rm tr}({\bf R}_k{\bf R}_i{\bf \Psi}_i^{-1}{\bf R}_i)}{{\rm tr}({\bf R}_i{\bf \Psi}_i^{-1}{\bf R}_i)}+\frac{{\rm tr}\left({\bf C}_k{\bf R}_k{\bf \Psi}_k^{-1}{\bf R}_k\right)}{{\rm tr}\left({\bf R}_k{\bf \Psi}_k^{-1}{\bf R}_k\right)},\\
\ddot N_k=\ &\left(\kappa_t^{\rm BS}+\kappa_r^{\rm UE}\right)P_t{\rm tr}({\bf R}_k)/M+\sigma_k^2,
\end{align}
and $\delta=M_{\rm E}/M$ denotes the normalized number of Eve's antennas. Notice that $R_{\rm sec}\big\vert_{q=0}$ in (35) decreases with $\delta$, and then the following proposition is provided by setting (35) to zero.
\begin{proposition}
In the absence of AN, we can compute the largest number of Eve's antennas $M_{\rm E}=\lfloor\delta_{\rm AN}M\rfloor$ for preserving a non-zero secrecy rate, where
\begin{align}
\delta_{\rm AN}=\frac{\ddot S_k\kappa_t^{\rm BS}\frac{K}{M}{\rm tr}({\bf Q}_{\rm E})}{\kappa_t^{\rm BS}\ddot S_k\frac{K{\rm tr}\left({\bf Q}_{\rm E}^2\right)}{{\rm tr}({\bf Q}_{\rm E})}+\frac{M}{\zeta}(\ddot I_k +\frac{K\ddot N_k}{P_t}) {\rm tr}\left({\bf R}_k{\bf \Psi}_k^{-1}{\bf R}_k{\bf Q}_{\rm E}\right)}.
\end{align}
\end{proposition}

This result shows that the presence of the transmit distortion noise is crucial for preserving a non-zero secrecy rate, especially when the transmitted AN power is zero, since $\kappa_t^{\rm BS}=0$ results in $\delta_{\rm AN}=0$.

For the general case in the presence of AN, the largest number of Eve's antennas that are allowed to preserve a positive secrecy rate is examined in the following proposition. In this regard, the ergodic secrecy rate in (34) is rewritten as (40), shown at the top of the next page, where
\begin{figure*}[ht]
	\begin{equation}
R_{\rm sec}=\left[\log_2\left(1+\frac{\xi \ddot S_k}{\xi\big(\ddot I_k+\frac{K}{M}{\rm tr}({\bf C}_k)\big)+\ddot D_k}\right)-\log_2\left(1+\frac{\xi \delta M^2{\rm tr}\left({\bf R}_k{\bf \Psi}_k^{-1}{\bf R}_k{\bf Q}_{\rm E}\right){\rm tr}({\bf Q}_{\rm E})/{\rm tr}({\bf R}_k{\bf \Psi}_k^{-1}{\bf R}_k)}{K\left(1-\xi+\kappa_t^{\rm BS}\right)\left[{\rm tr}({\bf Q}_{\rm E})\right]^2-\delta KM\ddot\chi{\rm tr}\left({\bf Q}_{\rm E}^2\right)}\right)\right]^+,
	\end{equation}
		\hrulefill
\end{figure*}
\begin{align}
\ddot D_k=\ &\frac{K}{M}\left({\rm tr}\left({\bf C}_k\right)+\left(\kappa_t^{\rm BS}+\kappa_r^{\rm UE}\right){\rm tr}({\bf R}_k)\right)+\frac{\sigma_k^2K}{P_t},\\
\ddot\chi=\ &\frac{\frac{(1-\xi)^2}{M-K}+2(1-\xi)\kappa_t^{\rm BS}+\left(\kappa_t^{\rm BS}\right)^2}{1-\xi+\kappa_t^{\rm BS}}.
\end{align}

\begin{proposition}
If Eve possesses less than $M_{\rm E}=\lfloor\delta_{\rm sec}M\rfloor$ antennas, then the legitimate user $k$ can leverage the AN to preserve a non-zero secrecy rate, where
\begin{align}
\delta_{\rm sec}=\frac{\ddot S_k K(1+\kappa_t^{\rm BS})\left[{\rm tr}({\bf Q}_{\rm E})\right]^2}{M^2\lambda_k\ddot D_k+\ddot S_k KM\frac{\frac{1}{M-K}+2\kappa_t^{\rm BS}+\left(\kappa_t^{\rm BS}\right)^2}{1+\kappa_t^{\rm BS}}{\rm tr}\left({\bf Q}_{\rm E}^2\right)}
\end{align}
and $\lambda_k=\frac{{\rm tr}\left({\bf R}_k{\bf \Psi}_k^{-1}{\bf R}_k{\bf Q}_{\rm E}\right)}{{\rm tr}({\bf R}_k{\bf \Psi}_k^{-1}{\bf R}_k)}{\rm tr}({\bf Q}_{\rm E})$.
\end{proposition}

It is observed that the largest number of Eve's antennas that may be accepted increases with the accuracy of the LMMSE estimation. Furthermore, $\delta_{\rm sec}$ decreases with the HWI parameters $\kappa_r^{\rm BS}$, $\kappa_r^{\rm UE}$, and $\kappa_t^{\rm UE}$, and increases with $\kappa_t^{\rm BS}$.

\subsection{New Characteristics of Integrating RIS}
The role of RIS for improving the secrecy rate is further analyzed in this subsection. Specifically, we characterize whether the favorable propagation conditions can be satisfied or not by employing a large number of RIS elements $N$ rather than a large number of BS antennas $M$. In order to tackle the challenge of determining the relationship between $R_{\rm sec}$ and $N$, we analyze the asymptotic behavior of $R_{\rm sec}$ in (34) when $M$ and $N$ are large.

\begin{proposition}
Assuming uncorrelated Rayleigh fading and ideal hardware for channel estimation, the ergodic secrecy rate reduces to $R_{\rm sec}=\left[R_k^{({\rm UR})}-\overline C_{\rm E}^{({\rm UR})}\right]^+$, where $R_k^{({\rm UR})}$ and $\overline C_{\rm E}^{({\rm UR})}$ are given in (44) and (45), shown at the top of this page, respectively, with
\begin{figure*}[ht]
	\begin{equation}
	R_k^{({\rm UR})}=\log_2\left(1+\frac{p\tau_u\rho{\rm tr}({\bf\Upsilon}_k)}   {p\tilde I_{k,i}+q\frac{M-K}{M}{\rm tr}\left(\beta_{2,k}{\bf I}_M+\beta_{{\rm I},k}{\bf H}_1{\bf H}_1^H-\tau_u\rho{\bf\Upsilon}_k\right)+\left(\kappa_t^{\rm BS}+\kappa_r^{\rm UE}\right)\frac{P_t}{M} {\rm tr}(\beta_{2,k}{\bf I}_M+\beta_{{\rm I},k}{\bf H}_1{\bf H}_1^H)+\sigma_k^2}\right),
	\end{equation}
	\hrulefill
\end{figure*}
\begin{figure*}[ht]
	\begin{equation}
	\overline C_{\rm E}^{({\rm UR})}=\log_2\left(1+\frac{pM_{\rm E}M\left(q(M-K)+\kappa_t^{\rm BS}P_t\right){\rm tr}(\beta_3{\bf I}_M+\beta_{\rm I,E}{\bf H}_1{\bf H}_1^H){\rm tr}\left((\beta_3{\bf I}_M+\beta_{\rm I,E}{\bf H}_1{\bf H}_1^H){\bf\Upsilon}_k\right)/{\rm tr}\left({\bf\Upsilon}_k\right)}  {\left(q(M-K)+\kappa_t^{\rm BS}P_t\right)^2\left[{\rm tr}(\beta_3{\bf I}_M+\beta_{\rm I,E}{\bf H}_1{\bf H}_1^H)\right]^2-\varpi{\rm tr}\left((\beta_3{\bf I}_M+\beta_{\rm I,E}{\bf H}_1{\bf H}_1^H)^2\right)}\right),
	\end{equation}
	\hrulefill
\end{figure*}
\begin{align}
{\bf\Upsilon}_k=\ &(\beta_{2,k}{\bf I}_M+\beta_{{\rm I},k}{\bf H}_1{\bf H}_1^H)\ddot{\bf \Psi}_k^{-1}(\beta_{2,k}{\bf I}_M+\beta_{{\rm I},k}{\bf H}_1{\bf H}_1^H),\\
\tilde I_{k,i}=\ &\sum\nolimits_{i\neq k}^K\frac{{\rm tr}\left((\beta_{2,k}{\bf I}_M+\beta_{{\rm I},k}{\bf H}_1{\bf H}_1^H){\bf\Upsilon}_i\right)}{{\rm tr}({\bf\Upsilon}_i)}\nonumber\\
&+\frac{{\rm tr}\left((\beta_{2,k}{\bf I}_M+\beta_{{\rm I},k}{\bf H}_1{\bf H}_1^H-\tau_u\rho{\bf\Upsilon}_k){\bf\Upsilon}_k\right)}{{\rm tr}\left({\bf\Upsilon}_k\right)},
\end{align}
and $\varpi=M_{\rm E}\left(\kappa_t^{\rm BS}P_t\right)^2+q^2M_{\rm E}M(M-K)+2qM_{\rm E}(M-K)\kappa_t^{\rm BS}P_t$.
\end{proposition}

The proof of Proposition 3 is obtained by setting ${\bf R}_{{\rm I}}={\bf I}_N$, ${\bf R}_{\rm B}={\bf I}_M$, and $\kappa_t^{\rm UE}=\kappa_r^{\rm BS}=0$ in Theorem~1 and Theorem~2. We see that the phase-shift matrix $\bf \Phi$ does not appear in $R_k^{({\rm UR})}$ and $\overline C_{\rm E}^{({\rm UR})}$. Thus, under spatially uncorrelated channels, a low-cost RIS with quantized phase shifts can be employed without hindering the ergodic secrecy rate. Moreover, it is noticed that the power of information signal in (44) asymptotically scales as $\mathcal O(M)$, while the AN power and the HWI power scale as $\mathcal O(1)$. The inter-user interference cannot be ignored when $M\rightarrow\infty$. This is because the asymptotic orthogonality among the legitimate users does not hold anymore due to the common component ${\bf H}_1$ of the cascaded channels.

Based on Proposition 3, we examine the benefits brought by a large number of RIS elements, i.e., in the asymptotic regime $N\rightarrow\infty$, which is provided in the following corollary.

\begin{corollary}
When $N\gg M$, ${\bf H}_1{\bf H}_1^H\rightarrow\beta_1N{\bf I}_M$ holds. Then, the ergodic secrecy rate in Proposition 3 simplifies to (48), shown at the top of the next page, where $\Xi_k=K(\beta_{2,k}+\beta_{{\rm I},k}\beta_1 N)-\bar\gamma_k$ and $\bar\gamma_k=\frac{(\beta_{2,k}+\beta_{{\rm I},k}\beta_1N)^2}{\beta_{2,k}+\beta_{{\rm I},k}\beta_1N+\frac{\sigma_u^2}{\tau_u\rho}}$.
\end{corollary}
\begin{figure*}[ht]
	\begin{align}
	R_{\rm sec}=\ &\left[\log_2\left(1+\frac{\xi P_t M \bar\gamma_k /K}{\xi P_t\Xi_k/K + (1-\xi)P_t \left(\beta_{2,k}+\beta_{{\rm I},k}\beta_1 N-\bar\gamma_k\right)+\left(\kappa_t^{\rm BS}+\kappa_r^{\rm UE}\right)P_t(\beta_{2,k}+\beta_{{\rm I},k}\beta_1 N)+\sigma_k^2}\right)\nonumber\right.\\
	&\left.-\log_2\left(1+\frac{pM_{\rm E}M \left(q(M-K)+\kappa_t^{\rm BS}P_t\right)(\beta_3+\beta_{\rm I,E}\beta_1N)^2}{M\left(q(M-K)+\kappa_t^{\rm BS}P_t\right)^2\left(\beta_3+\beta_{\rm I,E}\beta_1N\right)^2-\varpi(\beta_3^2+2\beta_3\beta_{\rm I,E}N+\beta_{\rm I,E}^2\beta_1^2N^2)}\right)\right]^{+},
	\end{align}
	\hrulefill
\end{figure*}

The obtained result shows that the RIS strengthens the channel gain through the parameter $\bar\gamma_k$, however, it results in worse estimation performance through the MSE matrix. As $N\rightarrow\infty$, the channel gain increases without bounds, yet the estimation error converges to $\frac{\sigma_u^2}{\tau_u\rho}$. Besides, we observe that the capacity of Eve remains roughly constant with $M$ and the achievable rate of the legitimate user $k$ increases logarithmically with $M$. This reveals that increasing $N$ significantly enhances the ergodic secrecy rate. Therefore, Corollary~3 demonstrates that under HWI and imperfect CSI, the ergodic secrecy rate can scale up as $\mathcal O(\log_2 (M))$ as a function of $M$.

The power scaling law with respect to $M$ has already been studied in traditional secure massive MIMO systems [43]. Interestingly, in the presence of an RIS, we provide an innovative power scaling law for the secure communication regarding $N$ in the following corollary.

\begin{corollary}
Assuming that the transmit power $P_t$ scales down at the rate $P_t=E_u/N$ with $N\rightarrow\infty$, the ergodic secrecy rate is expressed as (49), shown at the top of the next page.
\begin{figure*}[ht]
	\begin{align}
	R_{\rm sec}\rightarrow\ &\left[\log_2\left(1+\frac{\xi E_u M\beta_{{\rm I},k} \beta_1 /K}{\xi E_u (K-1)\beta_{{\rm I},k} \beta_1/K+\left(\kappa_t^{\rm BS}+\kappa_r^{\rm UE}\right)E_u\beta_{{\rm I},k} \beta_1+\sigma_k^2}\right)\nonumber\right.\\
	&\left.-\log_2\left(1+\frac{\xi M_{\rm E} M (1-\xi+\kappa_t^{\rm BS})/K }{M\left(1-\xi+\kappa_t^{\rm BS}\right)^2-M_{\rm E} \big[(\kappa_t^{\rm BS})^2+ M(1-\xi)^2/(M-K)+2(1-\xi)\kappa_t^{\rm BS}\big]}\right)\right]^+.
	\end{align}
	\hrulefill
\end{figure*}
\end{corollary}

\begin{remark}
Based on (49), important observations can be made. 1) Corollary~4 proves that we may achieve non-zero secrecy rates while reducing the total transmit power in a proportionate manner to the reciprocal of the number of RIS elements, i.e., $1/N$. This demonstrates the potential of the RIS providing unique robustness to HWI in securing massive MIMO systems. 2) The power scaling law is irrelevant to the large-scale fading coefficients of the direct links for a large number of $N$. 3) As $N$ grows to infinity, $1/N$ is the fastest declining rate of $P_t$ in order to preserve a positive rate of the legitimate user, while the capacity of Eve converges to a finite limit. We note that even without the AN, i.e., $\xi=1$, secure communications can be ensured by relying on the power scaling thanks to the transmit hardware distortion~$\kappa_t^{\rm BS}$.
\end{remark}

\begin{corollary}
Under the assumption of $M\gg K$, $M\gg M_{\rm E}$, and $N\rightarrow\infty$, the ergodic secrecy rate in (48) reduces to (50), shown at the top of the next page.
\begin{figure*}[ht]
	\begin{align}
R_{\rm sec}\rightarrow \ &\left[\log_2\left(1+\frac{\xi M /K}{\xi (K -1)/K+\kappa_t^{\rm BS}+\kappa_r^{\rm UE}}\right)-\log_2\left(1+\frac{\xi M_{\rm E}}{K(1-\xi+\kappa_t^{\rm BS})}\right)\right]^+.
	\end{align}
	\hrulefill
\end{figure*}
\end{corollary}

\begin{remark}
It is readily seen that the ergodic secrecy rate converges to a limit for large $N$, and the injection of AN is necessary for guaranteeing secure communications, especially without the transmit hardware distortion. The ergodic secrecy rate does not increase uniformly with $\kappa_t^{\rm BS}$, thus a large value of $\kappa_t^{\rm BS}$ can be harmful. We notice that the impact of the transmit HWI is decremental by enlarging the BS antenna array in a RIS-aided secure massive MIMO network.
\end{remark}

\subsection{Optimized Transmission Design}
In this subsection, we endeavor to adaptively adjust the power allocation coefficient $\xi$ to maximize the ergodic secrecy rate. From (40), we first reformulate the secrecy rate as
\begin{align}
R_{\rm sec}=\ &\left[\log_2\bigg(1+\frac{\xi \ddot S_k}{\xi\psi+\ddot D_k}\bigg)\nonumber\right.\\
&\left.-\log_2\left(1+\frac{\xi \upsilon A_1}{\upsilon^2 A_2-(1-\xi)^2A_3+\xi A_4-A_5}\right)\right]^+,
\end{align}
where $\psi=\ddot I_k+\frac{K}{M}{\rm tr}({\bf C}_k)$, $\upsilon=1-\xi+\kappa_t^{\rm BS}$, $A_1=M_{\rm E}M{\rm tr}({\bf R}_k{\bf \Psi}_k^{-1}{\bf R}_k{\bf Q}_{\rm E}){\rm tr}({\bf Q}_{\rm E})/{\rm tr}({\bf R}_k{\bf \Psi}_k^{-1}{\bf R}_k)$, $A_2=K\left[{\rm tr}({\bf Q}_{\rm E})\right]^2$, $A_3=\frac{M_{\rm E}MK}{M-K}{\rm tr}\left({\bf Q}_{\rm E}^2\right)$, $A_4=2M_{\rm E}K\kappa_t^{\rm BS}{\rm tr}\left({\bf Q}_{\rm E}^2\right)$, and $A_5=M_{\rm E}K\kappa_t^{\rm BS}(\kappa_t^{\rm BS}+2){\rm tr}\left({\bf Q}_{\rm E}^2\right)$. To proceed, the first derivative of $R_{\rm sec}$ in (51) is calculated by (52), shown at the top of this page.
\begin{figure*}[ht]
	\begin{align}
	\frac{\partial R_{\rm sec}}{\partial \xi}=\ &\frac{\ddot S_k\ddot D_k}{\ln2(\xi\psi+\ddot D_k)(\xi\psi+\ddot D_k+\xi \ddot S_k)}\nonumber\\
	&-\frac{A_1(1-2\xi+\kappa_t^{\rm BS})\left[\upsilon^2 A_2-(1-\xi)^2A_3+\xi A_4-A_5\right]-\xi\upsilon A_1\left[A_2(2\xi-2-2\kappa_t^{\rm BS})-A_3(2\xi-2)+A_4\right]}{\ln2\left[\upsilon^2 A_2-(1-\xi)^2A_3+\xi A_4-A_5\right]\left[\upsilon^2 A_2-(1-\xi)^2A_3+\xi A_4-A_5+\xi\upsilon A_1\right]}.
	\end{align}
\hrulefill
\end{figure*}
It is seen that the expression in (52) takes the form of high orders with respect to $\xi$. Thus, it is usually difficult to obtain analytical solutions directly. To address this, under the assumption of $M_{\rm E}K/M^2\ll1$, we deduce the optimal solution in closed form to get a desirable $\xi^*$. Numerical results are provided to assess its accuracy in the next section. Then, $\frac{\partial R_{\rm sec}}{\partial \xi}$ in (52) becomes
\begin{align}
R_{\rm sec}'=\frac{\partial R_{\rm sec}}{\partial \xi}=\ &\frac{\ddot S_k\ddot D_k}{\ln2(\xi\psi+\ddot D_k)(\xi\psi+\ddot D_k+\xi \ddot S_k)}\nonumber\\
	&-\frac{(1+\kappa_t^{\rm BS})A_1}{\ln2(\upsilon^2KL_1+\xi\upsilon A_1)},
\end{align}
where $L_1=[{\rm tr}({\bf Q}_{\rm E})]^2-\frac{M_{\rm E}M}{M-K}{\rm tr}\left({\bf Q}_{\rm E}^2\right)$. Based on the expression in (53), we further analyze the convexity of $R_{\rm sec}$ and then determine an optimal solution that satisfies the power constraints.
\begin{lemma}
	The optimal choice of the power allocation ratio $\xi$ is given by
	\begin{equation}
		\xi^*=\frac{b-\sqrt{b^2+4ac}}{2a},
	\end{equation}
	where $a=\ddot S_k\ddot D_k(L_1K-A_1)-(\psi^2+\psi\ddot S_k)(1+\kappa_t^{\rm BS})A_1$, $b=(1+\kappa_t^{\rm BS})[2\ddot S_k\ddot D_kL_1K-\ddot S_k\ddot D_kA_1+\ddot D_kA_1(2\psi+\ddot S_k)]$, and $c=(1+\kappa_t^{\rm BS})^2\ddot S_k\ddot D_kL_1K-(1+\kappa_t^{\rm BS})\ddot D_k^2A_1$ are constants over one coherence time.
\end{lemma}
\begin{proof}
By using the expression of $R'_{\rm sec}$ in (53), the second-order derivative $R''_{\rm sec}$ can be easily computed. Details are not provided here due to the page limit. Through the analysis, we prove that $R'_{\rm sec}$ exhibits a strictly decreasing behavior regarding $\xi$ since $R''_{\rm sec}<0$. Further, it is verified that $R'_{\rm sec}$ is decreasing from positive to negative values over the range of the considered intervals. Thus, $R_{\rm sec}$ is convex, which implies that the unique optimal value of $\xi$ exists. It is obtained directly by solving $R'_{\rm sec}=0$.
\end{proof}

\section{Numerical Results}
In this section, simulation results are provided to validate the theoretical analysis of RIS-assisted downlink secure communication in the presence of HWI and imperfect CSI. Specifically, we consider the following setup: all the legitimate users and Eve are evenly dispersed around a 50-meter-radius circle, the center of which is situated 100 and 200 meters away from the RIS and the BS, respectively. The LoS channel model between the BS and the RIS is expressed as
\begin{align}
[\bar{\bf H}_1]_{a,b}=\ &\exp\left(j\frac{2\pi}{\lambda}((a-1)d_{\rm BS}\sin(\theta_1(b))\sin(\phi_1(b))   \right.\nonumber\\
&\left. +(b-1)d_{\rm RIS}\sin(\theta_2(a))\sin(\phi_2(a)))\right),
\end{align}
where $d_{\rm BS} = 0.5\lambda$ and $d_{\rm RIS} = 0.25\lambda$ are the antenna spacing of the BS and the RIS, respectively. In addition, $\theta_1(b)$ and $\phi_1(b)$ are generated by uniform distributions between 0 to $\pi$ and 0 to $2\pi$, which further yields $\theta_2(a) = \pi-\theta_1(b)$ and $\phi_2(a) = \pi +\phi_1(b)$. The large-scale path loss of the reflected and the direct links are denoted by $\beta_r=J_0\left({d}/{J_1}\right)^{-\zeta_r}$ and $\beta_d=J_0\left({d}/{J_1}\right)^{-\zeta_d}$, where $\zeta_r$ and $\zeta_d$ are the path loss exponent of the corresponding links, and $J_0=-20$ dB accounts for the path loss at a distance of $J_1=1$ m. Regarding the spatial correlation, the exponential correlation model (i.e., $[{\bf A}]_{ij}=l^{\vert i-j\vert}$) is adopted at the BS, where $0\leq l<1$ stands for the antenna correlation index. The pilot length is equal to $\tau_u = K$ and the RIS phase shifts are set to $\pi/4$. Unless otherwise stated, the simulation parameters are given in Table~I for convenience.

\begin{table}[t]
\centering
\caption{\text {Simulation Parameters}}
\begin{tblr}{
  vline{2} = {-}{},
  hline{1,8} = {-}{0.08em},
  hline{2} = {-}{},
}
Parameters                                & Values           \\
Path loss exponent for the RIS-aided link & $\zeta_r=2.1$       \\
Path loss exponent for direct link        & $\zeta_d=3.2$       \\
Correlation coefficient at the BS         & $l=0.6$          \\
Spacing among RIS elements         & $d_{\rm H}=d_{\rm V}=\lambda/2$ \\
RIS phase noise power                     & $\sigma_p^2=0.1$    \\
Power allocation coefficient              & $\xi=0.5$      
\end{tblr}
\end{table}

\begin{figure}[t]
	\centering
	\centerline{\includegraphics[width=0.4\textwidth]{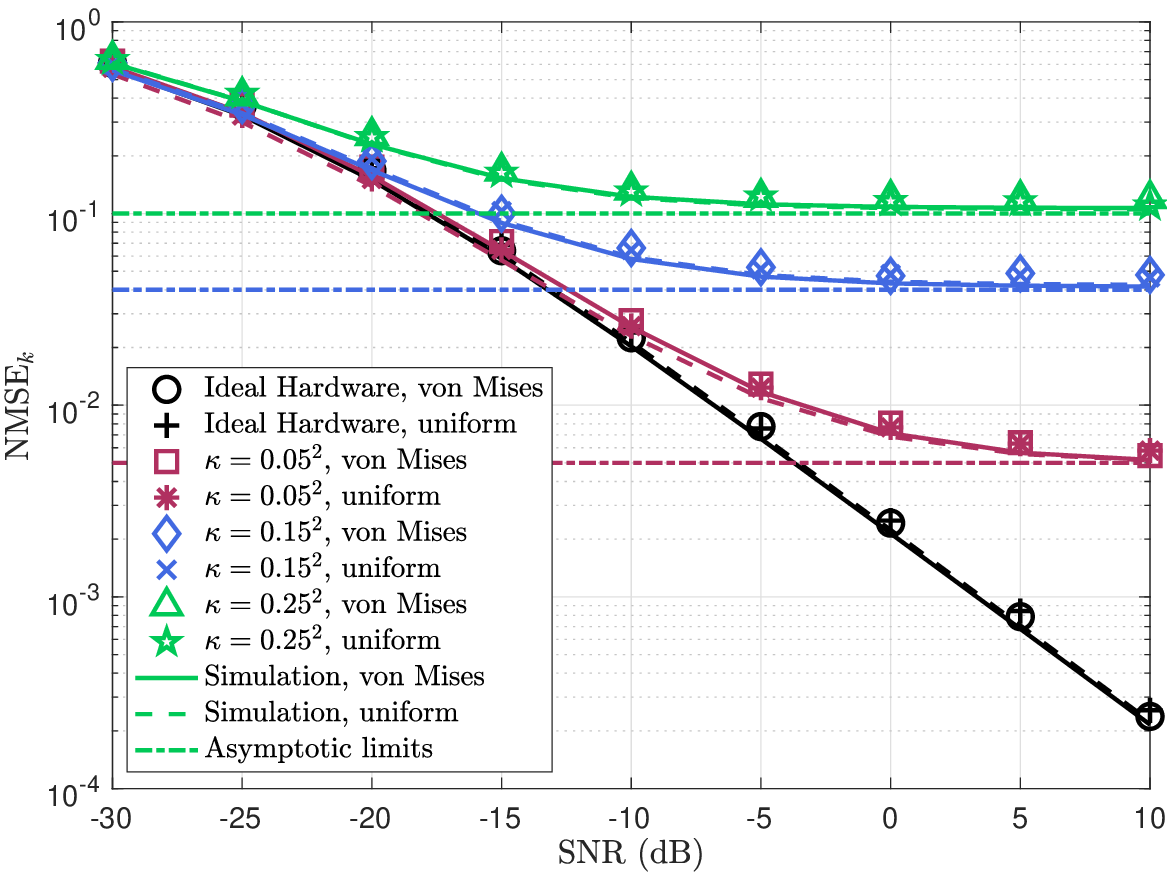}}
	\caption{NMSE versus SNR with different transceiver hardware parameters.}
\end{figure}

\begin{figure}[t]
	\centering
	\centerline{\includegraphics[width=0.4\textwidth]{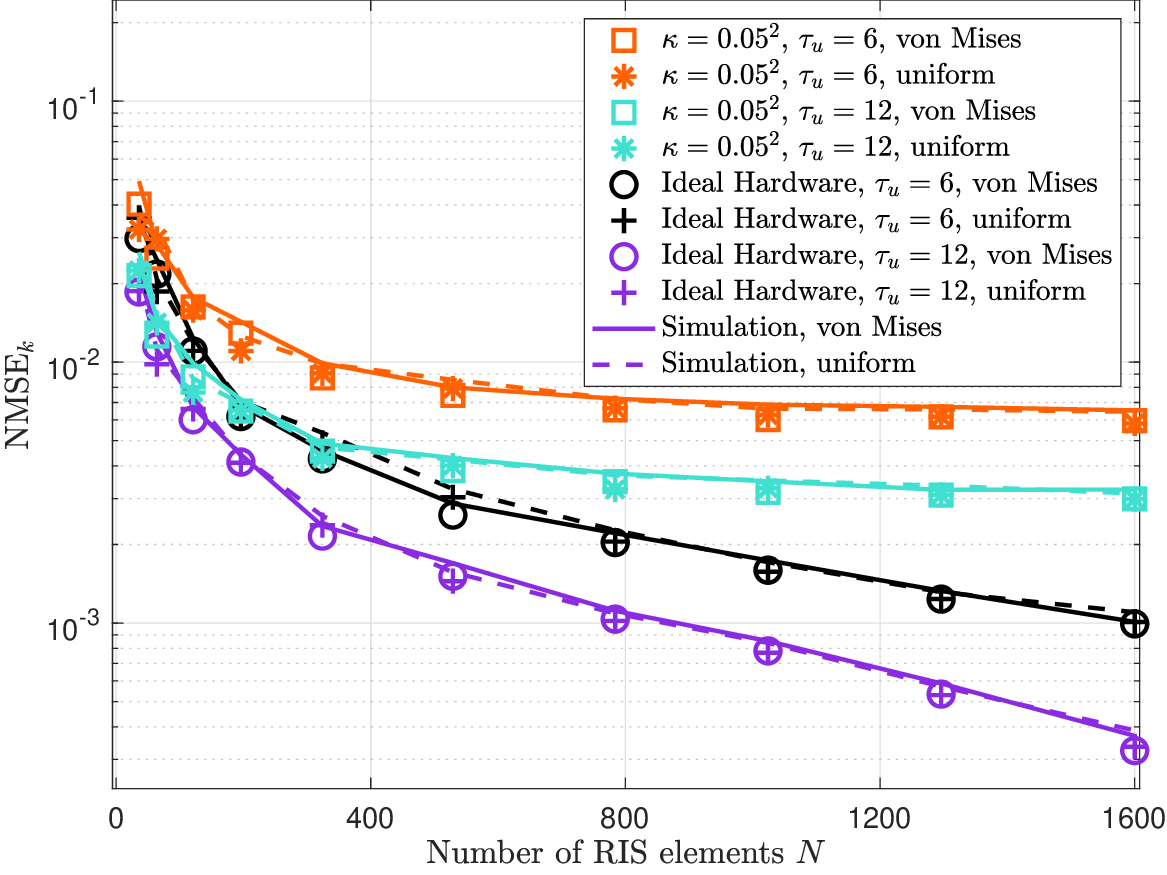}}
	\caption{NMSE versus the number of RIS elements with different transceiver hardware parameters and pilot lengths.}
\end{figure}

\begin{figure}[t]
	\centering
	\centerline{\includegraphics[width=0.4\textwidth]{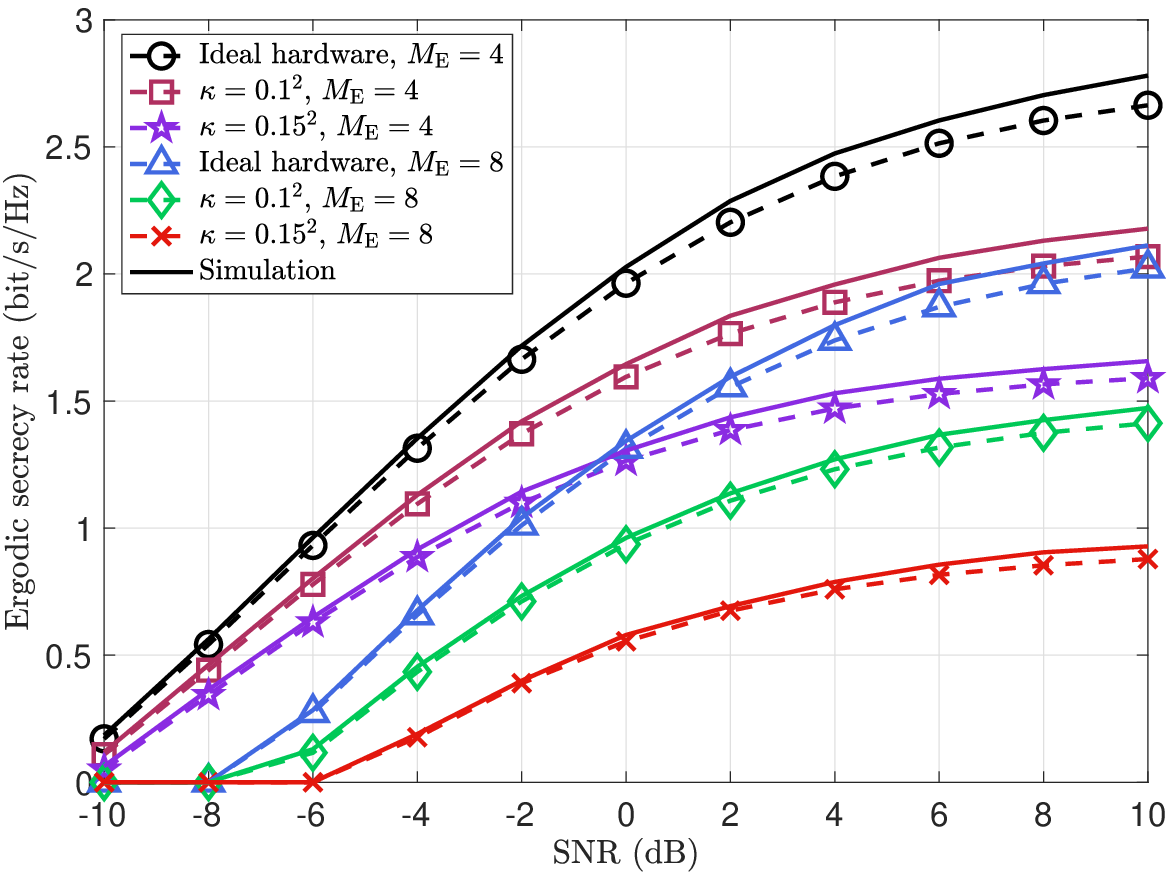}}
	\caption{Ergodic secrecy rate and lower bound versus SNR ($M=128$, $N=196$, $K=6$).}
\end{figure}

\begin{figure}[t]
	\centering
	\centerline{\includegraphics[width=0.4\textwidth]{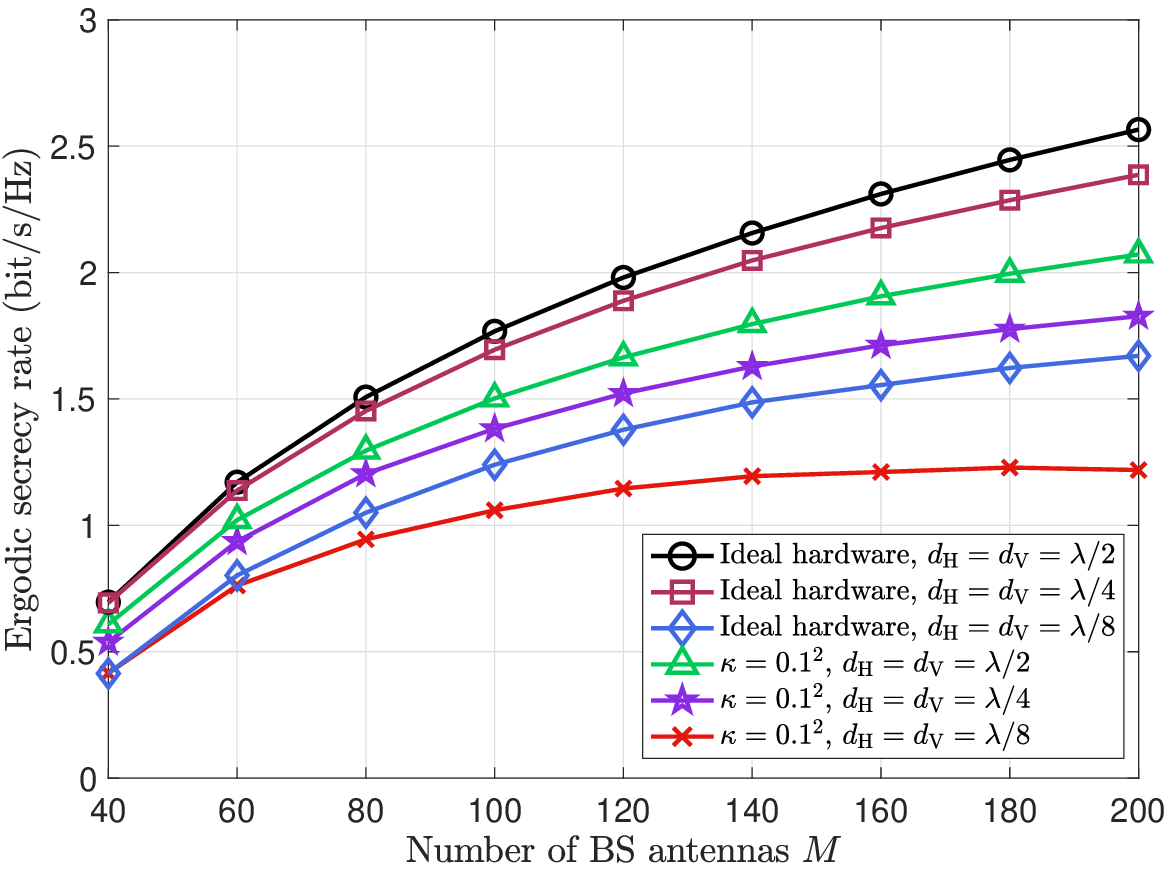}}
	\caption{Ergodic secrecy rate versus the number of BS antennas ($N=196$, $K=6$, $M_{\rm E}=4$, SNR$=0$ dB).}
\end{figure}

\begin{figure}[t]
	\centering
	\centerline{\includegraphics[width=0.4\textwidth]{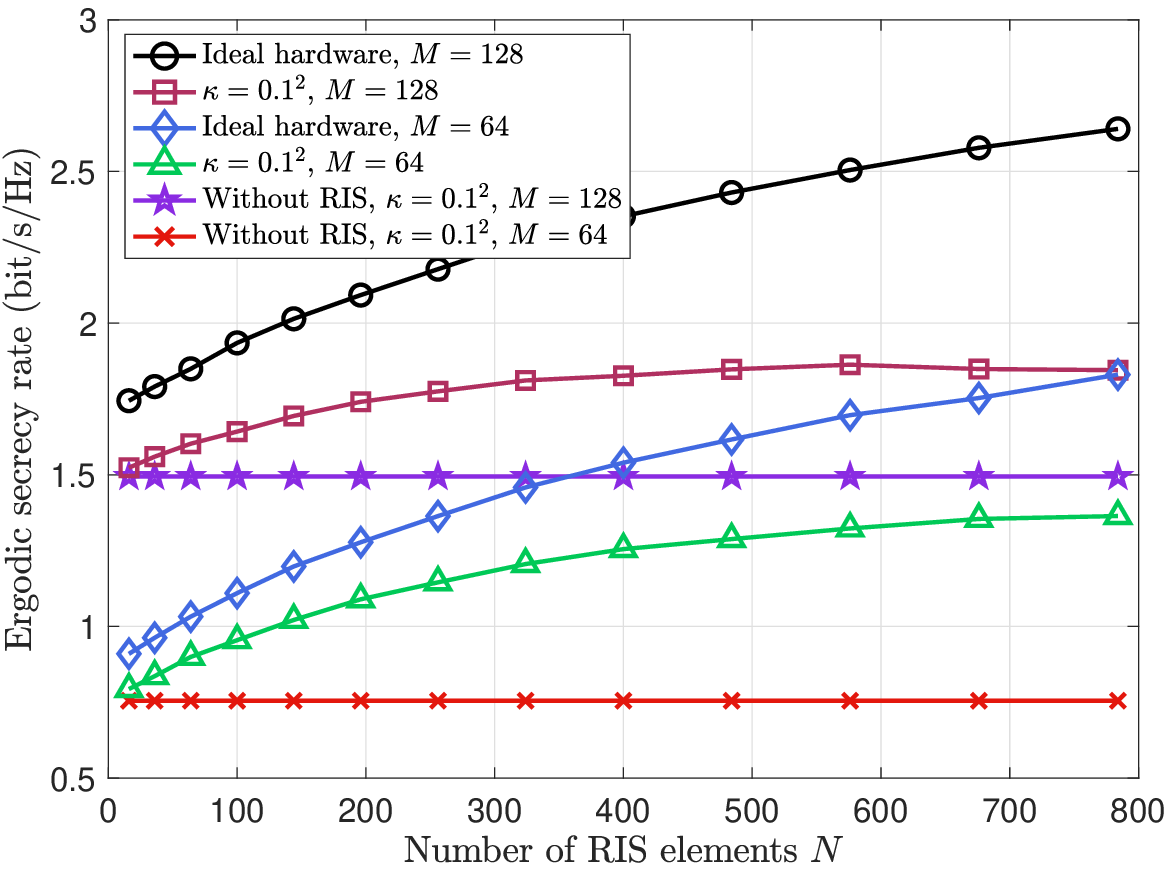}}
	\caption{Ergodic secrecy rate versus the number of RIS elements ($M=64$, $K=6$, $M_{\rm E}=4$, SNR$=0$ dB).}
\end{figure}

\begin{figure}[t]
	\centering
	\centerline{\includegraphics[width=0.4\textwidth]{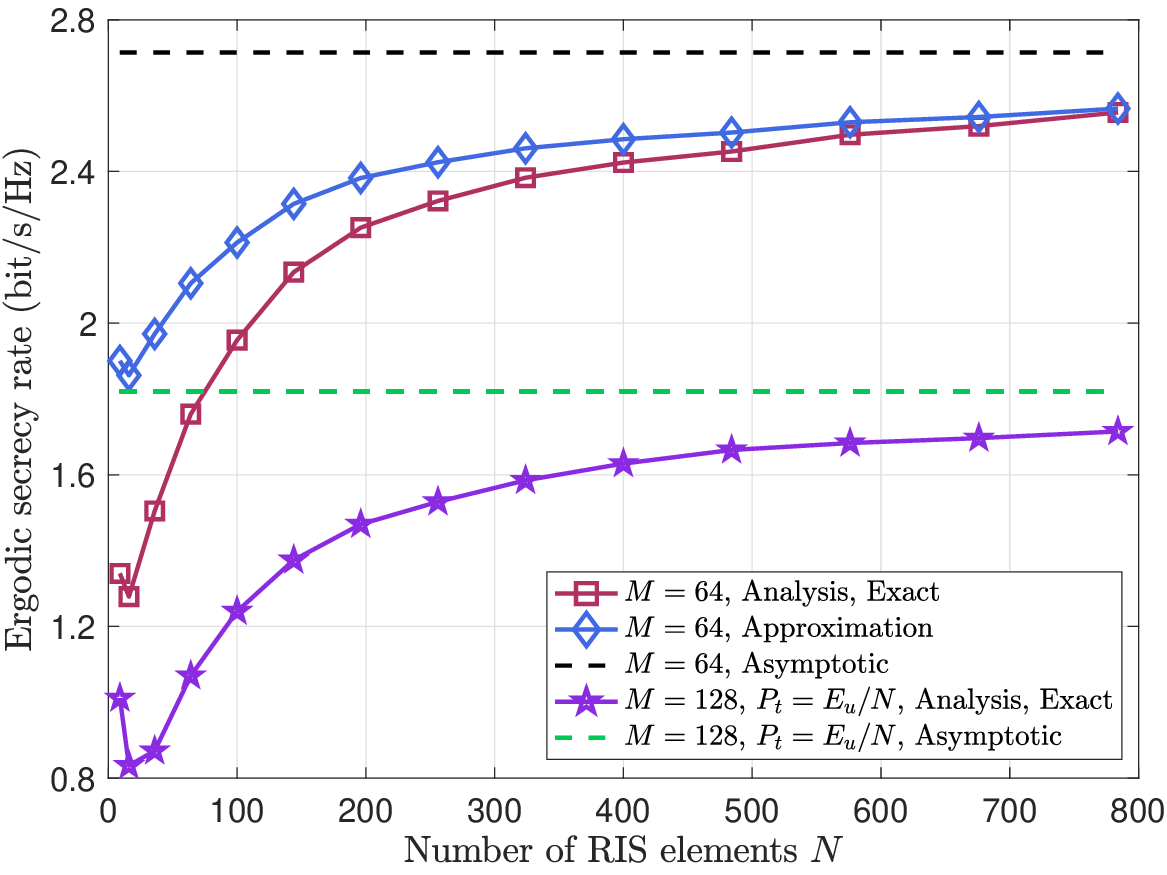}}
	\caption{Ergodic secrecy rate versus the number of RIS elements without RIS phase-shift error ($E_u=20$ dB, $K=6$, $M_{\rm E}=4$).}
\end{figure}

\begin{figure}[t]
	\centering
	\centerline{\includegraphics[width=0.4\textwidth]{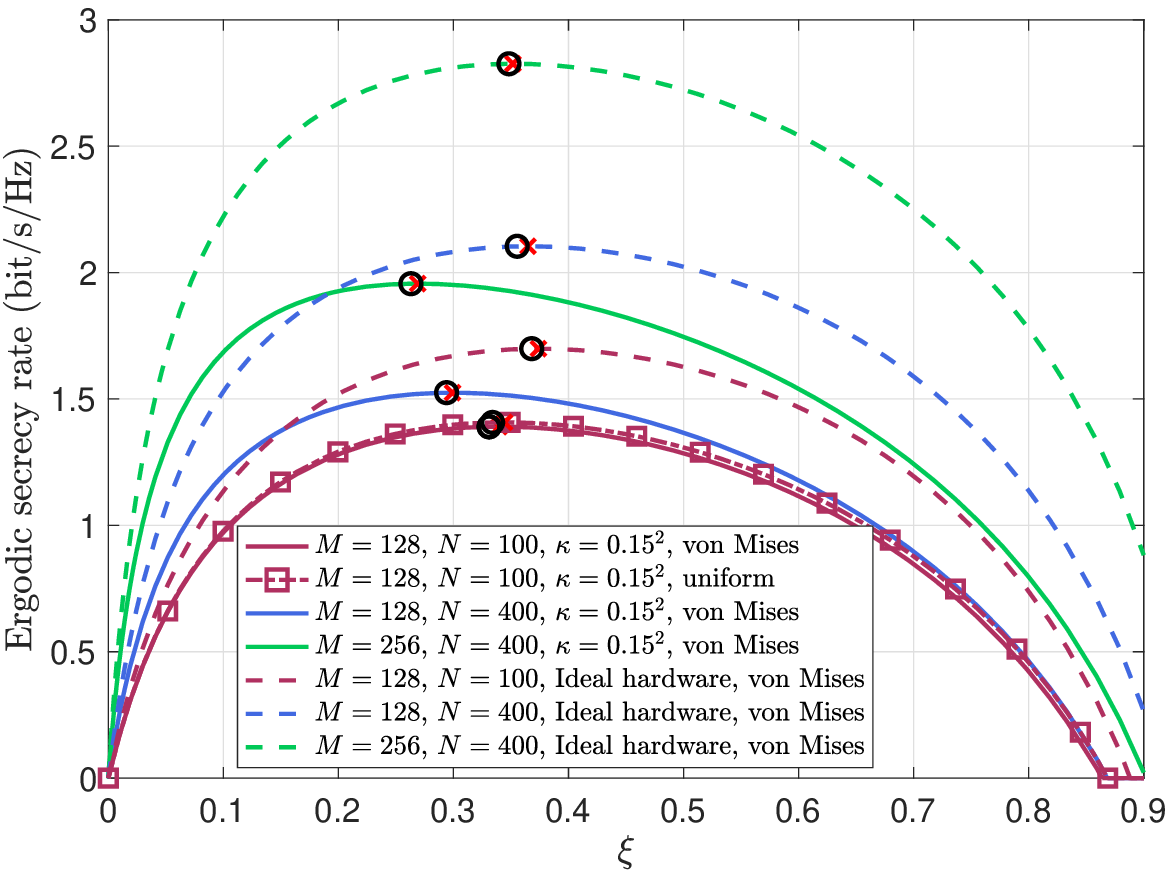}}
	\caption{Ergodic secrecy rate versus $\xi$ with different transceiver hardware parameters ($K=10$, $M_{\rm E}=4$, SNR$=0$ dB).}
\end{figure}

\begin{figure}[t]
	\centering
	\centerline{\includegraphics[width=0.4\textwidth]{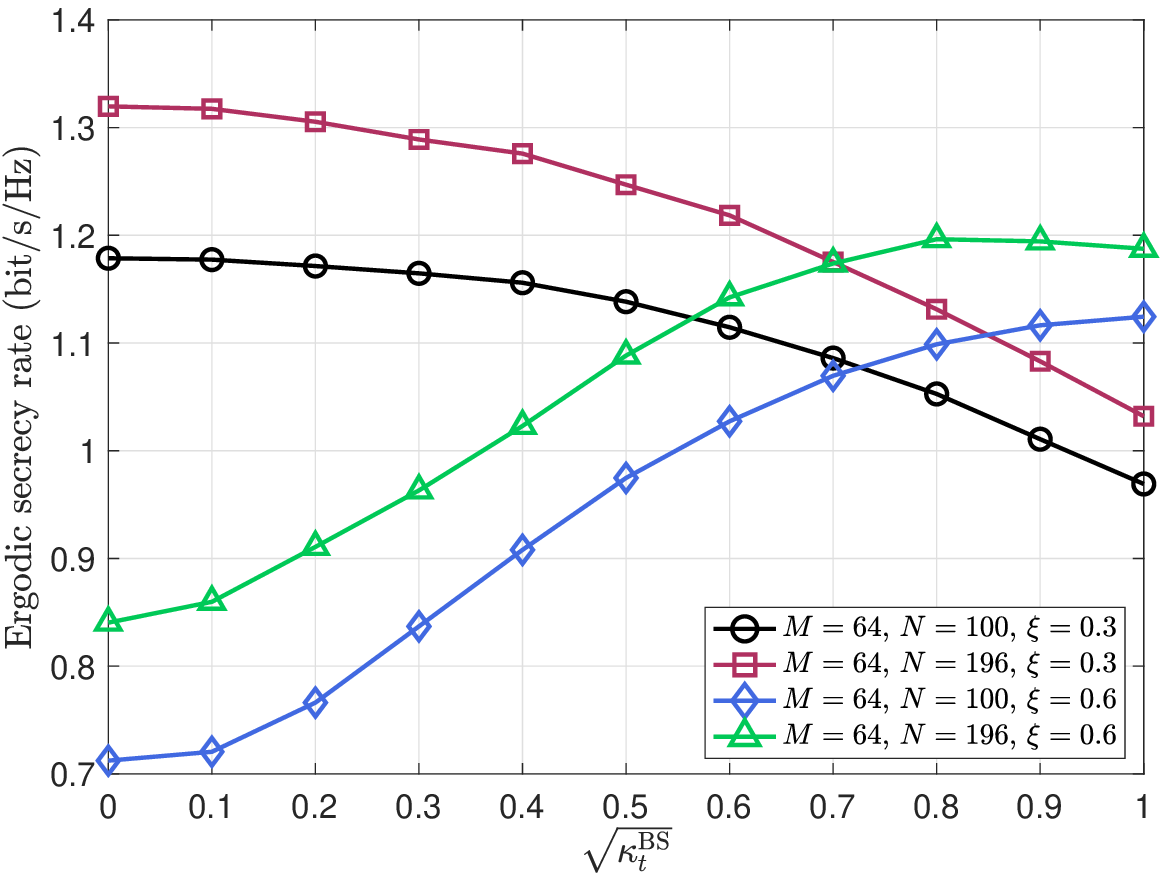}}
	\caption{Ergodic secrecy rate versus the transmit distortion parameter $\kappa_t^{\rm BS}$ ($M_{\rm E}=4$, $\kappa_r^{\rm UE}=\kappa_t^{\rm UE}=\kappa_r^{\rm BS}=0.1^2$, SNR$=0$ dB).}
\end{figure}

\begin{figure}[t]
	\centering
	\centerline{\includegraphics[width=0.4\textwidth]{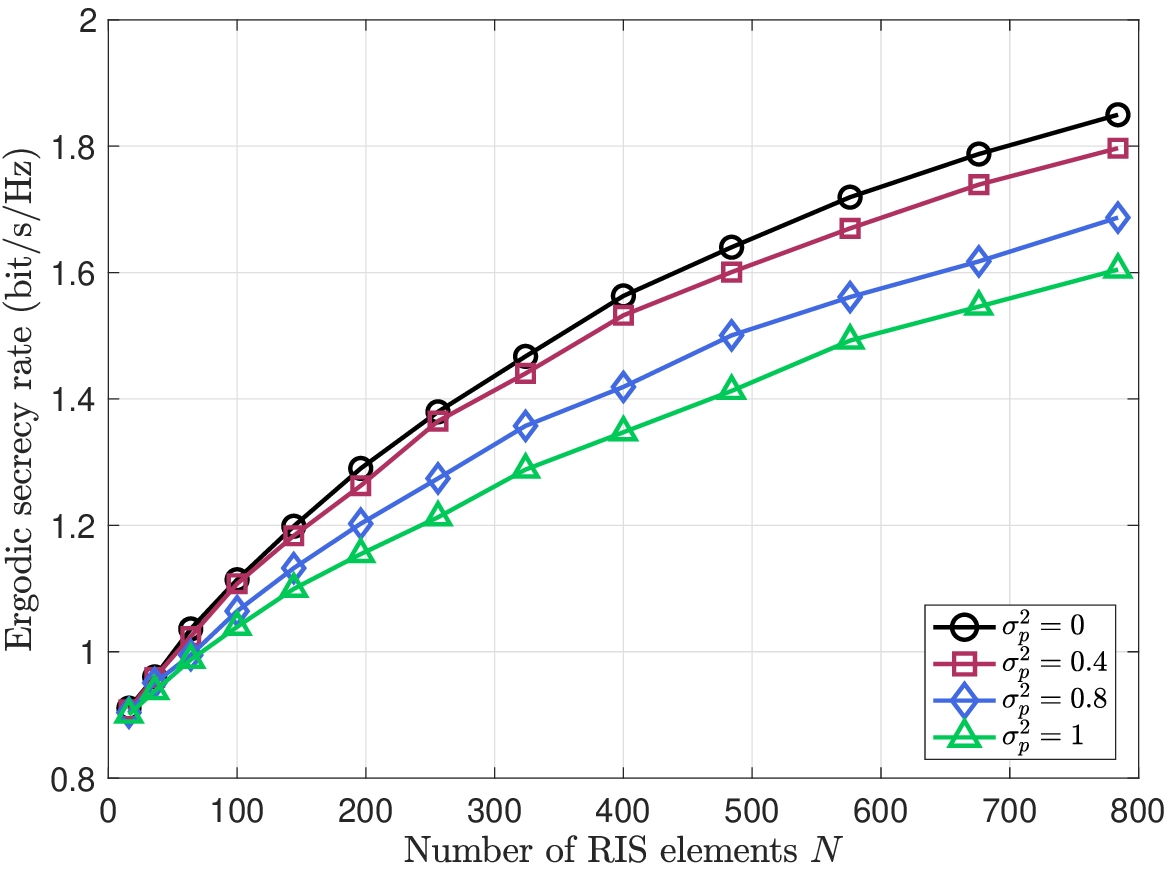}}
	\caption{Ergodic secrecy rate versus the number of RIS elements with different RIS phase noise power $\sigma_p^2$ ($M=64$, $M_{\rm E}=4$, SNR$=0$ dB).}
\end{figure}

Fig. 2 depicts the NMSE performance of the LMMSE channel estimator versus SNR for different values of HWI parameters $\kappa$. The results for ideal hardware (i.e., $\kappa_t^{\rm UE} = \kappa_r^{\rm BS}= \kappa=0$) are also considered for comparison. As predicted in Corollary~1, at high values of $\rho$, the NMSE$_k$ approaches certain error floors under the constraint of imperfect hardware, while the NMSE$_k$ decreases without bounds with ideal hardware. Besides, we see that the NMSE$_k$ is nearly identical for uniform and von Mises distributions of the RIS phase-shift error. This is because the RIS assigns the same phase noise power in both cases.

Fig. 3 plots the NMSE performance versus the number of RIS elements $N$ for different pilot lengths. It is evident that the NMSE$_k$ decreases with $N$ and approaches zero in the asymptotic regime of $N\rightarrow\infty$, which coincides with Corollary~2. The reason can be explained as follows. As $N$ increases, the RIS-aided channels of legitimate users become dominated by boosted channel gain. Additionally, a large number of $N$ leads to a more deterministic channel between the BS and the RIS, resulting in small estimation errors. Moreover, we observe that extending the pilot sequence can attain better estimation performance of the LMMSE channel estimator as expected from Section~III-A.

Fig. 4 illustrates the ergodic secrecy rate versus SNR with $\kappa_t^{\rm UE} = \kappa_r^{\rm BS}= \kappa_r^{\rm UE} = \kappa_t^{\rm BS}=\kappa$. From Fig. 4, we observe that the analytical expression in (34) and the Monte Carlo simulation match well, which confirms the accuracy of the derivations. Moreover, a large number of $M_{\rm E}$ and a high level of the transceiver HWI jeopardize the ergodic secrecy rate due to the improved wiretapping ability at Eve and severely impaired signal at the legitimate users. Especially, in the regime of high SNR, the impact of HWI becomes more pronounced since the power of HWI is dominated in (26).

We next analyze the ergodic secrecy rate versus the number of BS antennas $M$ in Fig. 5 for different values of $\kappa$. Clearly, as $M$ grows, the ergodic secrecy rate increases without bounds in the case of ideal hardware. However, a secrecy performance gap between ideal hardware and $\kappa=0.1^2$ appears and then expands with $M$. In particular, the impact of spatial correlation at the RIS is also evaluated. It is observed that the ergodic secrecy rate degrades significantly when the spatial correlation increases, i.e., the values of $d_{\rm H}$ and $d_{\rm V}$ decrease from $\lambda/2$ to $\lambda/8$. This is because the degrees of freedom shrink with smaller interdistance among RIS elements. Furthermore, the ergodic secrecy rate under different values of $\kappa $ and $d_{\rm H}$/$d_{\rm V}$ is very close when the number of $M$ is small. In contrast, for large $M$, the secrecy performance gain becomes more prominent with ideal hardware.

We evaluate the ergodic secrecy rate versus the number of RIS elements $N$ in Fig. 6 for different parameter settings. It can be found that the ergodic secrecy rate in the case of the ideal hardware grows without bounds as $N$ increases, similarly to the behavior of increasing $M$. However, when the number of $N$ is large, the ergodic secrecy rate saturates quickly for $\kappa=0.1^2$. In addition, by comparing the ergodic secrecy rate with and without RIS, the advantages of an RIS are confirmed even with preconfigured RIS phase shifts. 

Fig. 7 examines the ergodic secrecy rate versus $N$ when ideal hardware for channel estimation and spatially uncorrelated channels are considered. The exact analytical ergodic secrecy rate is plotted by using the expression in Proposition~3, and the approximated and the asymptotic secrecy rates are plotted by using the expressions in (48) and (50), respectively. First, we can see that the asymptotic secrecy rate remains constant with $N$, and the exact secrecy rate approaches the theoretical limit in (50) for large numbers of $N$. Besides, the performance gap between the approximated and the exact analytical secrecy rates declines with $N$ as expected. Furthermore, the derived power scaling laws in (49) are demonstrated where $P_t$ scales down as $1/N$. We observe that in the cases of imperfect CSI and HWI, the ergodic secrecy rate can preserve positive values via the proposed power scaling in the regime of $N\rightarrow\infty$, which verifies the results in Corollary~4.

Fig. 8 shows the relation between the ergodic secrecy rate and the power allocation coefficient $\xi$. The numerical and analytical solutions of the optimal $\xi$ are provided by using markers $\times$ and $\circ$, respectively. It can be found that $\xi^*$ in (54) depends on the parameters $M$, $N$, and $\kappa$. Specifically, $\xi^*$ is a monotonically decreasing function of $M$, while there is a slight shift of $\xi^*$ as $N$ increases from $100$ to $400$. In other words, the BS needs to increase the transmitted power of the AN when the number of $M$, $N$ increases. This can be explained as follows. On one hand, the correlation between the received signals at Eve and legitimate users always leads to degradations of the ergodic secrecy rate since they share the same channel component ${\bf H}_1$. On the other hand, the channel distinctions due to the presence of ${\bf H}_1$ exhibit low for large $M$ and $N$, while the degrees of freedom of the RIS-user/Eve channel is magnified as $N$ increases. Thus, the information leakage becomes more severe for large $M$ rather than large $N$. However, when the value of $\xi$ is larger than $\xi^*$, the secrecy performance gain by increasing $N$ is diminished for $\kappa=0.15^2$ due to insufficient interference at Eve. Moreover, $\xi^*$ is nearly identical under different distributions of the RIS phase-shift error since the same phase noise power contributes equally to the secrecy performance.

In Fig. 9, we examine the impact of the transmit distortion parameter $\kappa_t^{\rm BS}$ on the ergodic secrecy rate under different setups ($\xi$, $M$, and $N$). Intuitively, a large value of $\kappa_t^{\rm BS}$ is not always harmful to the ergodic secrecy rate. More specifically, with sufficient AN power (i.e., for small values of $\xi$), the ergodic secrecy rate is negatively affected by the transmit HWI. In contrast, for large values of $\xi$, an increasing ergodic secrecy rate is observed since the impact of transmit HWI is similar to that of AN. Moreover, the secrecy rate loss that comes from $\kappa_t^{\rm BS}$ can be compensated by increasing $N$, while the secrecy performance gain produced by large $N$ is diminished as $\kappa_t^{\rm BS}$ increases. This is because the transmit HWI dominates the secrecy performance in which the benefits of the RIS become inconspicuous.

In Fig. 10, the impact of RIS phase noise on the ergodic secrecy rate is evaluated versus the number of RIS elements $N$ when the ideal transceiver hardware is considered. We observe that the ergodic secrecy rate deteriorates as the RIS phase noise power $\sigma_p^2$ increases. Moreover, when the RIS phase noise power $\sigma_p^2=1$, the number of $N$ that are needed to achieve the same target secrecy rate performance of $R_{\rm sec} = 1.6$ bit/s/Hz is nearly doubled compared to the case when the phase noise power at the RIS is $\sigma_p^2=0$.

\section{Conclusion}
This paper presented an investigation of the ergodic secrecy rate in RIS-aided multiuser massive MIMO with a passive Eve while taking CSI imperfections, transceiver HWI, RIS phase noise, and correlated Rayleigh fading into account. By utilizing the obtained channel estimation results, we deduced a theoretical lower bound for the ergodic secrecy rate at the target user with MRT and the null-space AN scheme. Our theoretical analysis proved that the total transmit power may decline proportionally to $1/N$, yet it can preserve non-zero secrecy rates in the asymptotic regime of $N\rightarrow\infty$. Besides, in the presence of imperfect CSI and HWI, we showed that the ergodic secrecy rate on the order of $\mathcal O(\log_2 (M))$ can be attained as $N$ goes to infinity, and the impact of transceiver HWI is decremental when $M$ is large. Furthermore, we found that the ergodic secrecy rate increases significantly with $N$ regardless of the power allocation coefficient, however, under imperfect hardware, the advantage brought by the RIS becomes marginal.

\section*{Appendix A}
Given that ${\bf g}_{{\rm I},k}$ and ${\bf g}_{{\rm B},k}$ are independent, the distribution of ${\bf h}_k$ can be determined by ${\bf h}_k\sim\mathcal{CN}(0, {\bf R}_k)$, where ${\bf R}_k={\bf R}_{{\rm B},k}+{\bf H}_1{\bf \Phi}\widetilde{\bf R}_{{\rm I},k}{\bf \Phi}^H{\bf H}_1^H$ with
\begin{align}
\widetilde{\bf R}_{{\rm I},k}=\mathbb E\{{\bf \Theta}{\bf R}_{{\rm I},k}{\bf \Theta}^H\}
=\varrho^2{\bf R}_{{\rm I},k}+\beta_{{\rm I},k}(1-\varrho^2){\bf I}_N
\end{align}
and $\varrho=\mathbb E\{e^{j\tilde{\theta}_n}\}$. In cases where the RIS phase-shift error obeys the uniform distribution, we obtain $\varrho=\frac{\sin(\iota_p)}{\iota_p}$; in cases where it follows the von Mises distribution, we obtain $\varrho=\frac{I_1(\nu_p)}{I_0(\nu_p)}$.

According to the properties of the LMMSE channel estimation scheme, we have
\begin{align}
\hat{\bf h}_k=\mathbb E\{{\bf y}_{p,k}{\bf h}_k^H\}\left(\mathbb E\left\{{\bf y}_{p,k}{\bf y}_{p,k}^H\right\}\right)^{-1}{\bf y}_{p,k}.
\end{align}
The first term $\mathbb E\{{\bf y}_{p,k}{\bf h}_k^H\}$ is obtained as
\begin{align}
&\mathbb E\{{\bf y}_{p,k}{\bf h}_k^H\}\nonumber\\
&=\mathbb E\bigg\{\bigg(\tau_u\sqrt{\rho}{\bf h}_k+\sum\limits_{i=1}^{K}{\boldsymbol \eta}_{t,i}^H{\boldsymbol\phi}_k{\bf h}_i+{\boldsymbol \Upsilon}_r^{\rm BS}{\boldsymbol\phi}_k+{\bf N}_p{\boldsymbol\phi}_k\bigg){\bf h}_k^H\bigg\}\nonumber\\
&=\tau_u\sqrt{\rho}\mathbb E\left\{{\bf h}_k{\bf h}_k^H\right\}=\tau_u\sqrt{\rho}{\bf R}_k.
\end{align}
Then, the term $\mathbb E\{{\bf y}_{p,k}{\bf y}_{p,k}^H\}$ is calculated as
\begin{align}
&\mathbb E\left\{{\bf y}_{p,k}{\bf y}_{p,k}^H\right\}\nonumber\\
&=\tau_u^2\rho\mathbb E\left\{{\bf h}_k{\bf h}_k^H\right\}+\mathbb E\left\{\sum\limits_{i=1}^{K}{\boldsymbol \eta}_{t,i}^H{\boldsymbol\phi}_k{\bf h}_i\big(\sum\limits_{i=1}^{K}{\boldsymbol \eta}_{t,i}^H{\boldsymbol\phi}_k{\bf h}_i\big)^H\right\}\nonumber\\
&\ \ \ +\mathbb E\{{\boldsymbol \Upsilon}_r^{\rm BS}{\boldsymbol\phi}_k({\boldsymbol \Upsilon}_r^{\rm BS}{\boldsymbol\phi}_k)^H\}+\mathbb E\{{\bf N}_p{\boldsymbol\phi}_k({\bf N}_p{\boldsymbol\phi}_k)^H\}\nonumber\\
&=\tau_u^2\rho{\bf R}_k+\tau_u\rho\kappa_t^{\rm UE}\sum_{i=1}^{K}{\bf R}_i+\tau_u\rho\kappa_r^{\rm BS}\sum_{i=1}^{K}{\bf I}_M\circ{\bf R}_i+\tau_u\sigma_u^2{\bf I}_M.
\end{align}
By substituting (58) and (59) into (57), the proof is completed.

\section*{Appendix B}
Let us compute the expectations in (24) one by one.

1) For the numerator in (25), we have
\begin{align}
\left\vert{\mathbb E}\{{\bf h}_k^H{\bf w}_k\}\right\vert^2=\ &\frac{1}{\mathbb E\{\Vert \hat{\bf h}_k \Vert^2\}}\big\vert{\mathbb E}\{{\bf h}_k^H\hat{\bf h}_k\}\big\vert^2\nonumber\\
=\ &\frac{1}{\mathbb E\{\Vert \hat{\bf h}_k \Vert^2\}}\big\vert{\mathbb E}\{\hat{\bf h}_k^H\hat{\bf h}_k\}+{\mathbb E}\{{\bf e}_k^H\hat{\bf h}_k\}\big\vert^2\nonumber\\
=\ & \tau_u\rho{\rm tr}({\bf R}_k{\bf \Psi}_k^{-1}{\bf R}_k),
\end{align}
where $\mathbb E\{\Vert \hat{\bf h}_k \Vert^2\}=\tau_u\rho{\rm tr}({\bf R}_k{\bf \Psi}_k^{-1}{\bf R}_k)$ as well as the independence of $\hat{\bf h}_k$ and the channel estimation error are exploited.

2) The power of the interference in (25) is computed as
\begin{align}
{\mathbb E}\{|{\bf h}_k^H{\bf w}_i|^2\}=\ &\frac{{\mathbb E}\{|{\bf h}_k^H\hat{\bf h}_i|^2\}}{\tau_u\rho{\rm tr}({\bf R}_i{\bf \Psi}_i^{-1}{\bf R}_i)}\nonumber\\
=\ &\frac{{\mathbb E}\{|\hat{\bf h}_k^H\hat{\bf h}_i|^2\}+{\mathbb E}\{|{\bf e}_k^H\hat{\bf h}_i|^2\}}{\tau_u\rho{\rm tr}({\bf R}_i{\bf \Psi}_i^{-1}{\bf R}_i)}\nonumber\\
=\ &\frac{{\rm tr}({\bf R}_k{\bf R}_i{\bf \Psi}_i^{-1}{\bf R}_i)}{{\rm tr}({\bf R}_i{\bf \Psi}_i^{-1}{\bf R}_i)},
\end{align}
due to the independence of ${\bf e}_k$ and $(\hat{\bf h}_k,\hat{\bf h}_i)$ and the zero mean of the channel estimation error.

3) The power of the signal uncertainty is calculated as
\begin{align}
&\mathbb V{\rm ar}\left\{{\bf h}_k^H{\bf w}_k\right\}\nonumber\\
&=\frac{1}{\mathbb E\{\Vert \hat{\bf h}_k \Vert^2\}}\bigg(\underbrace{\mathbb E\big\{\vert{\bf h}_k^H\hat{\bf h}_k\vert^2\big\}}_{\mathcal I_1}-\underbrace{\big\vert \mathbb E\{{\bf h}_k^H\hat{\bf h}_k\} \big\vert^2}_{\mathcal I_2}\bigg),
\end{align}
where the expectation $\mathcal I_1$ is derived as
\begin{align}
&\mathcal I_1= \mathbb E\left\{\vert\hat{\bf h}_k^H\hat{\bf h}_k\vert^2\right\}+\mathbb E\left\{\vert{\bf e}_k^H\hat{\bf h}_k\vert^2\right\}\nonumber\\
&= \rho^2\mathbb E\left\{\vert{\bf y}_{p,k}^H{\bf\Psi}_k^{-1}{\bf R}_k{\bf R}_k{\bf\Psi}_k^{-1}{\bf y}_{p,k}\vert^2\right\}+{\rm tr}\left(\mathbb E\{{\bf e}_k{\bf e}_k^H\hat{\bf h}_k\hat{\bf h}_k^H\}\right) \nonumber\\
&= \tau_u^2\rho^2\left\vert {\rm tr}({\bf R}_k{\bf \Psi}_k^{-1}{\bf R}_k) \right\vert^2+\tau_u\rho{\rm tr}\left({\bf C}_k{\bf R}_k{\bf \Psi}_k^{-1}{\bf R}_k\right)
\end{align}
and the expectation $\mathcal I_2$ equals
\begin{align}
	\mathcal I_2=\ &\rho^2\left\vert \mathbb E\{{\bf y}_{p,k}^H{\bf\Psi}_k^{-1}{\bf R}_k{\bf R}_k{\bf\Psi}_k^{-1}{\bf y}_{p,k}\} \right\vert^2\nonumber\\
	=\ &\tau_u^2\rho^2\left\vert {\rm tr}({\bf R}_k{\bf \Psi}_k^{-1}{\bf R}_k) \right\vert^2.
\end{align}
Combining (62), (63), and (64), we obtain the variance as
\begin{align}
	\mathbb V{\rm ar}\left\{{\bf h}_k^H{\bf w}_k\right\}=\frac{{\rm tr}\left({\bf C}_k{\bf R}_k{\bf \Psi}_k^{-1}{\bf R}_k\right)}{{\rm tr}\left({\bf R}_k{\bf \Psi}_k^{-1}{\bf R}_k\right)}.
\end{align}

4) The power of AN leakage that the legitimate user~$k$ has received is computed~as
\begin{align}
{\mathbb E}\left\{{\bf h}_k^H{\bf V}{\bf V}^H{\bf h}_k\right\}=\ &{\mathbb E}\left\{\hat{\bf h}_k^H{\bf V}{\bf V}^H\hat{\bf h}_k\right\}+{\mathbb E}\left\{{\bf e}_k^H{\bf V}{\bf V}^H{\bf e}_k\right\}\nonumber\\
=\ &\sum\nolimits_{i=1}^{M-K}{\mathbb E}\left\{{\rm tr}\left({\bf v}_i^H{\bf e}_k{\bf e}_k^H{\bf v}_i\right)\right\}\nonumber\\
=\ &\frac{M-K}{M}{\rm tr}\left({\bf C}_k\right),
\end{align}
where we have leveraged the feature for the null-space AN method that $\hat{\bf h}_k^H{\bf V}={\bf 0}$ holds and the independence of the channel estimation error and the AN precoding matrix.

5) The power of the transceiver HWI is computed as follows
\begin{align}
{\mathbb E}\left\{{\bf h}_k^H{\bf\Upsilon}_t^{\rm BS}{\bf h}_k\right\}=\ &\kappa_t^{\rm BS}\frac{P_t}{M}{\rm tr}({\bf R}_k),\\
\mathbb E\{\mu_{r,k}^2\}=\ &\kappa_r^{\rm UE}\frac{P_t}{M}{\rm tr}({\bf R}_k).
\end{align}
By substituting (60), (61), and (65)-(68) into (24), the proof is completed.

\section*{Appendix C}
By means of the Jensen's inequality, the ergodic capacity of Eve can be computed by
\begin{align}
	C&\le\log_2\left(1+\mathbb{E}\{\gamma_{\rm E}\}\right).
\end{align}
It is obvious that the independence between ${\bf f}_{k,{\rm E}}$ and ${\bf H}_{\rm E}^H\left(q{\bf V}{\bf V}^H+{\bf \Upsilon}_t^{\rm BS}\right){\bf H}_{\rm E} \triangleq {\bf X}$ is essential towards further simplifying $\gamma_{\rm E}$. We note that ${\bf X}$ is undoubtedly independent of ${\bf f}_{k,{\rm E}}$ since it converges to a deterministic matrix when $M \rightarrow\infty$. Consequently, $\mathbb E\left\{({\bf f}_{k,{\rm E}})^H {\bf X}^{-1}{\bf f}_{k,{\rm E}}\right\}$ can be reformulated as $\mathbb E\left\{({\bf f}_{k,{\rm E}})^H \mathbb E\left\{{\bf X}^{-1}\right\}{\bf f}_{k,{\rm E}}\right\}$, where $\mathbb E\left\{{\bf X}^{-1}\right\}$ is roughly characterized as a scaled identity matrix. Based on (2), the distribution of ${\bf H}_{\rm E}$ is still Gaussian with covariance matrix ${\bf Q}_{\rm E}\triangleq{\bf R}_{\rm B,E}+{\bf R}_{\rm E}$, where ${\bf R}_{\rm E}={\bf H}_1{\bf \Phi}\widetilde{\bf R}_{\rm I, E}{\bf \Phi}^H{\bf H}_1^H$ with $\widetilde{\bf R}_{\rm I, E}=\varrho^2{\bf R}_{\rm I,E}+\beta_{\rm I,E}(1-\varrho^2){\bf I}_N$.

Then, we evince that the matrix ${\bf \Upsilon}_t^{\rm BS}$ converges to a deterministic scaled identity matrix $\kappa_t^{\rm BS}\frac{P_t}{M}{\bf I}_M$ due to the strong law of large numbers. Hence, ${\bf X}$ is expressed as
\begin{align}
	{\bf X}
	=\ &{\bf H}_{\rm E}^H\left(q{\bf V}{\bf V}^H+\kappa_t^{\rm BS}\frac{P_t}{M}{\bf I}_M\right){\bf H}_{\rm E}\nonumber\\
	=\ &q{\bf H}_{\rm E}^H{\bf V}{\bf V}^H{\bf H}_{\rm E}+\kappa_t^{\rm BS}\frac{P_t}{M}{\bf H}_{\rm E}^H[{\bf V}\ {\bf V}_0][{\bf V}\ {\bf V}_0]^H{\bf H}_{\rm E}\nonumber\\
	=\ &\left(q+\kappa_t^{\rm BS}\frac{P_t}{M}\right){\bf H}_{\rm E}^H{\bf V}{\bf V}^H{\bf H}_{\rm E}+\kappa_t^{\rm BS}\frac{P_t}{M}{\bf H}_{\rm E}^H{\bf V}_0{\bf V}_0^H{\bf H}_{\rm E},
\end{align}
where we have used $[{\bf V}\ {\bf V}_0][{\bf V}\ {\bf V}_0]^H={\bf I}_M$. The accurate distribution of ${\bf X}$ is hard to derive, and ${\bf X}$ is not, strictly speaking, a Wishart matrix [13], [44]. Nevertheless, by utilizing the moment matching method,  the distribution of ${\bf X}$ can be approximately represented as one scaled Wishart matrix, i.e., ${\bf X}\sim \varphi \mathcal{W}_{M_{\rm E}}(\eta, {\bf I}_{M_{\rm E}})$, in which the parameters $\varphi$ and $\eta$ are selected such that both sides in (70) have the identical first two moments [45]. To achieve this, eigendecompose ${\bf Q}_{\rm E}={\bf U} {\bf\Delta} {\bf U}^H$ to decorrelate the aggregate channel matrix ${\bf H}_{\rm E}$ as ${\bf \Omega}={\bf H}_{\rm E}{\bf\Delta}^{-1/2}{\bf U}^H$ with ${\bf\Delta}={\rm diag}(\lambda_1, \dots ,\lambda_m, \dots, \lambda_M)$ holding the eigenvalues of ${\bf Q}_{\rm E}$ and the columns of $\bf U$ being the associated eigenvectors. It is obvious that the distributions of $\bf \Omega U$ and $\bf \Omega$ are the same, given that $\bf U$ is unitary. Consequently, the distributions of ${\bf H}_{\rm E}^H{\bf V}{\bf V}^H{\bf H}_{\rm E}$ and ${\bf H}_{\rm E}^H{\bf V}_0{\bf V}_0^H{\bf H}_{\rm E}$ are equal to
\begin{equation}
		\sum_{i=1}^{M}\sum_{j=1}^{M}\lambda_i^{\frac{1}{2}}\lambda_j^{\frac{1}{2}}{\boldsymbol \omega}_i^H{\bf v}_{i}{\bf v}_{j}^H{\boldsymbol \omega}_j
\end{equation}
and
\begin{equation}
		\sum_{i=1}^{M}\sum_{j=1}^{M}\lambda_i^{\frac{1}{2}}\lambda_j^{\frac{1}{2}}{\boldsymbol \omega}_i^H{\bf v}_{0,i}{\bf v}_{0,j}^H{\boldsymbol \omega}_j,
\end{equation}
where ${\boldsymbol \omega}_i$ is the $i$-th row of ${\bf \Omega}$, and ${\bf v}_i$ and ${\bf v}_{0,i}$ are the $i$-th rows of ${\bf V}$ and ${\bf V}_0$, respectively.
Based on the independence of ${\boldsymbol \omega}_i$ and ${\boldsymbol \omega}_j$, we obtain that ${\bf H}_{\rm E}^H{\bf V}{\bf V}^H{\bf H}_{\rm E}$ follows a Wishart distribution $\sum_{l=1}^{M}\lambda_l\mathcal{W}_{M_{\rm E}}\left(M-K, \frac{1}{M}{\bf I}_{M_{\rm E}}\right)$ and ${\bf H}_{\rm E}^H{\bf V}_0{\bf V}_0^H{\bf H}_{\rm E}$ also follows $\sum_{l=1}^{M}\lambda_l\mathcal{W}_{M_{\rm E}}\left(K, \frac{1}{M}{\bf I}_{M_{\rm E}}\right)$ [46]. Therefore, equating the first and the second moments of those matrices leads to
\begin{align}
	\eta\varphi=\frac{{\rm tr}({\bf Q}_{\rm E})}{M}\left[\left(q+\kappa_t^{\rm BS}\frac{P_t}{M}\right)(M-K)+K\kappa_t^{\rm BS}\frac{P_t}{M}\right]
\end{align}
and
\begin{align}
\eta\varphi^2=\frac{{\rm tr}({\bf Q}_{\rm E}^2)}{M}\left[\left(q+\kappa_t^{\rm BS}\frac{P_t}{M}\right)^2(M-K)+K\left(\kappa_t^{\rm BS}\frac{P_t}{M}\right)^2 \right],
\end{align}
where we have used the properties that $\sum_{m=1}^{M}\lambda_m={\rm tr}({\bf Q}_{\rm E})$ and $\sum_{m=1}^{M}\lambda_m^2={\rm tr}({\bf Q}_{\rm E}^2)$. Then, the parameters $\varphi$ and $\eta$ are computed as
\begin{align}
	\varphi=\frac{{\rm tr}\left({\bf Q}_{\rm E}^2\right)[q^2(M-K)+2q\kappa_t^{\rm BS}P_t\frac{M-K}{M}+(\kappa_t^{\rm BS}P_t)^2/M]}{{\rm tr}({\bf Q}_{\rm E})[q(M-K)+\kappa_t^{\rm BS}P_t]}
\end{align}
and
\begin{align}
	\eta=\frac{\left[{\rm tr}({\bf Q}_{\rm E})\right]^2[q(M-K)+\kappa_t^{\rm BS}P_t]^2/M}{{\rm tr}\left({\bf Q}_{\rm E}^2\right)[q^2(M-K)+2q\kappa_t^{\rm BS}P_t\frac{M-K}{M}+(\kappa_t^{\rm BS}P_t)^2/M]}.
\end{align}

Subsequently, $\mathbb{E}\{\gamma_{\rm E}\}$ is given by
\begin{align}
	p\mathbb E\left\{({\bf f}_{k,{\rm E}})^H {\bf X}^{-1}{\bf f}_{k,{\rm E}}\right\}\ &\overset{\rm (a)}{=}\frac{p}{\varphi(\eta-M_{\rm E})}\mathbb E\left\{({\bf f}_{k,{\rm E}})^H{\bf f}_{k,{\rm E}}\right\}
	\nonumber\\
	&\overset{\rm (b)}{=}\frac{p M_{\rm E}{\rm tr}\left({\bf R}_k{\bf \Psi}_k^{-1}{\bf R}_k{\bf Q}_{\rm E}\right)}{\varphi(\eta-M_{\rm E}){\rm tr}({\bf R}_k{\bf \Psi}_k^{-1}{\bf R}_k)},
\end{align}
where $\rm (a)$ uses the property of obeying the distribution ${\bf \Xi}\sim\mathcal{W}_b(a,{\bf I}_b)$, i.e., ${\bf \Xi}^{-1}\xrightarrow{{\rm a.s.}}1/(a-b){\bf I}_b$ with $a>b$ [46, Sec. 2.1.6], and $\rm (b)$ results from
\begin{align}
\mathbb E\left\{({\bf f}_{k,{\rm E}})^H{\bf f}_{k,{\rm E}}\right\}=\ &\mathbb{E}\big\{{\bf w}_k^H{\bf H}_{\rm E}{\bf H}_{\rm E}^H{\bf w}_k\big\}\nonumber\\
=\ &M_{\rm E}\mathbb{E}\big\{{\bf w}_k^H{\bf Q}_{\rm E}{\bf w}_k\big\}\nonumber\\
=\ &\frac{M_{\rm E}{\rm tr}\left({\bf R}_k{\bf \Psi}_k^{-1}{\bf R}_k{\bf Q}_{\rm E}\right)}{{\rm tr}({\bf R}_k{\bf \Psi}_k^{-1}{\bf R}_k)}.
\end{align}
This finally yields the upper bound in (29) by substituting (75), (76), and (78) into (77).

\end{document}